\newcommand{\typeof}{1} %

\RequirePackage{ifthen}

\newcommand{\longv}[1]{\ifthenelse{\equal{\typeof}{1}}{#1}{}}
\newcommand{\shortv}[1]{\ifthenelse{\equal{\typeof}{0}}{#1}{}}
\newcommand{\shortlongv}[2]{\ifthenelse{\equal{\typeof}{0}}{#1}{#2}}

\shortlongv{
  \documentclass[a4paper,UKenglish]{lipics}
}{
  \documentclass[a4paper,UKenglish,10pt,DIV=9]{scrartcl}
}

\usepackage{microtype}
\usepackage{textcomp}
\usepackage{latexsym} 
\usepackage{amsfonts,amssymb,amsmath} 
\usepackage[sort&compress,numbers]{natbib}
\usepackage{hyperref}
\usepackage{enumitem}
\usepackage{tikz}
\usepackage{boxedminipage}
\usepackage{xparse}
\usetikzlibrary{fit}
\usepackage{forest}
\usepackage{general}
\usepackage{paper}

\shortlongv{
  \theoremstyle{plain}
  \newtheorem{proposition}[theorem]{Proposition}
  \theoremstyle{remark}
  \newtheorem{claim}[]{Claim}
}
{
  \usepackage{subfig}
  \usepackage{amsthm}
  \newtheorem{theorem}{Theorem}[section]
  
  \newtheorem{lemma}[theorem]{Lemma}
  \newtheorem{claim}[theorem]{Claim}
  \newtheorem{proposition}[theorem]{Proposition}
  \newtheorem{definition}[theorem]{Definition}
  \newtheorem{remark}[theorem]{Remark}
  \newtheorem{example}[theorem]{Example}
}

\newenvironment{varitemize}
{\begin{list}{\labelitemii}
{%
\setlength{\itemsep}{0pt}%
 \setlength{\topsep}{0pt}%
 \setlength{\parsep}{0pt}%
 \setlength{\partopsep}{0pt}%
 \setlength{\leftmargin}{15pt}%
 \setlength{\rightmargin}{0pt}%
 \setlength{\itemindent}{0pt}%
 \setlength{\labelsep}{5pt}%
 \setlength{\labelwidth}{10pt}}}
{\end{list}}

\newcounter{numberone}

\newenvironment{varenumerate}
{\begin{list}{\textcolor{darkgray}{\sffamily\bfseries\arabic{numberone}.}}
{\usecounter{numberone}
  \setlength{\itemsep}{0pt}
  \setlength{\topsep}{0pt}
  \setlength{\parsep}{0pt}
  \setlength{\partopsep}{0pt}
  \setlength{\leftmargin}{15pt}
  \setlength{\rightmargin}{0pt}
  \setlength{\itemindent}{0pt}
  \setlength{\labelsep}{5pt}
  \setlength{\labelwidth}{15pt}
}}
{\end{list}}

\newenvironment{framed}[0]{\begin{boxedminipage}{\linewidth}\vspace{-3mm}}{\end{boxedminipage}\vspace{-2mm}}

\usepackage[textsize=tiny]{todonotes}
\newcounter{comment}

\shortlongv{
\title{On Sharing, Memoization, and Polynomial Time\footnote{This work was partially supported by FWF project number J\ 3563.}}
\titlerunning{On Sharing, Memoization and Polynomia Time}

\author{Martin Avanzini}
\author{Ugo {Dal Lago}}
\affil{Dipartimento di Informatica - Scienza e Ingegneria, Universit\`a di Bologna, Italy
  \texttt{martin.avanzini@uibk.ac.at} and \texttt{dallago@cs.unibo.it}}
\authorrunning{M. Avanzini and U. Dal Lago} 
\Copyright{Martin Avanzini and Ugo Dal Lago}

\serieslogo{}
\volumeinfo
  {Billy Editor and Bill Editors}
  {2}
  {Conference title on which this volume is based on}
  {1}
  {1}
  {1}
\EventShortName{}
\DOI{10.4230/LIPIcs.xxx.yyy.p}
}{
\title{On Sharing, Memoization, and Polynomial Time\footnote{This work was partially supported by FWF project number J\ 3563.}\\({Long Version})}
\author{Martin Avanzini \and Ugo Dal Lago}
}

\begin{document}
\maketitle
\begin{abstract}
  We study how the adoption of an evaluation mechanism with sharing
  and memoization impacts the class of functions which can be computed
  in polynomial time. We first show how a natural cost model in
  which lookup for an already computed value has no cost is indeed
  invariant. As a corollary, we then prove that the most general
  notion of ramified recurrence is sound for polynomial time, this
  way settling an open problem in implicit computational complexity.
\end{abstract}
\longv{\tableofcontents\newpage}
\section{Introduction}
Traditionally, complexity classes are defined by giving bounds on the
amount of resources algorithms are allowed to use while solving
problems. This, in principle, leaves open the task of understanding
the \emph{structure} of complexity classes. As an example, a
given class of functions is not necessarily closed under composition or, more
interestingly, under various forms of recursion. When the class under
consideration is not too large, say close enough to the class of
\emph{polytime computable functions}, closure under recursion does not hold:
iterating over an efficiently computable function is not necessarily
efficiently computable, e.g.\ when the iterated function grows more than
linearly. In other words, characterizing complexity classes by purely
recursion-theoretical means is non-trivial.

In the past twenty years, this challenge has been successfully tackled,
by giving \emph{restricted} forms of recursion for which not only
certain complexity classes are closed, but which \emph{precisely}
generate the class. This has been proved for classes like \PTIME,
\PSPACE, the polynomial hierarchy \PH, or even smaller ones like \NC\
(more information about related work is in Section~\ref{sect:rw}). A
particularly fruitful direction has been the one initiated by
Bellantoni and Cook, and independently by Leivant, which consists in
restricting the primitive recursive scheme by making it
\emph{predicative}, thus forbidding those nested recursive definitions
which lead outside the classes cited above. Once this is settled, one
can tune the obtained scheme by either adding features (e.g.\ parameter
substitutions) or further restricting the scheme (e.g.\ by way of linearization).

Something a bit disappointing in this field is that the expressive
power of the simplest (and most general) form of predicative
recurrence, namely \emph{simultaneous} recurrence on \emph{generic
  algebras} is unknown.  If algebras are restricted to be
\emph{string} algebras, or if recursion is not simultaneous, soundness
for polynomial time computation is known to
hold~\cite{Leivant:FMII:95,DMZ:DICE:10}. The two soundness results are
obtained by quite different means, however: in presence of trees, one
is forced to handle \emph{sharing}~\cite{DMZ:DICE:10} of common
sub-expressions, while simultaneous definitions by recursion requires
a form of \emph{memoization}~\cite{Leivant:FMII:95}.

In this paper, we show that sharing and memoization can indeed be
reconciled, and we exploit both to give a new invariant time cost
model for the evaluation of rewrite systems. That paves the way
towards a polytime soundness for simultaneous predicative recursion
on generic algebras, thus solving the open problem we were mentioning.
More precisely, with the present paper we make the following
contributions:
\begin{varenumerate}
\item 
  We define a simple functional programming language.  The domain of
  the defined functions is a free algebra formed from constructors.
  Hence we can deal with functions over strings, lists, but also trees
  (see Section~\ref{s:basics}).  We then extend the underlying
  rewriting based semantics with \emph{memoization},
  i.e.\ intermediate results are automatically tabulated to avoid
  expensive re-computation (Section~\ref{s:invariance}).  As standard
  for functional programming languages such as \tool{Haskell} or
  \tool{OCaml}, data is stored in a \emph{heap}, facilitating
  \emph{sharing} of common sub-expression.  To measure the
  \emph{runtime} of such programs, we employ a novel cost model,
  called \emph{memoized runtime complexity}, where each function
  application counts one time unit, but lookups of tabulated calls do
  not have to be accounted.
\item 
  Our \emph{invariance theorem} (see Theorem~\ref{t:invariance})
  relates, within a polynomial overhead, the memoized runtime
  complexity of programs to the cost of implementing the defined
  functions on a classical model of computation, e.g. \emph{Turing} or 
  \emph{random access machines}. The invariance theorem thus confirms that
  our cost model truthfully represents the computational complexity 
  of the defined function.
\item
  We extend upon Leivant's notion of \emph{ramified recursive
    functions}~\cite{Leivant:POPL:93} by allowing definitions by
  \emph{generalised ramified simultaneous recurrence} (\emph{GRSR} for
  short).  We show that the resulting class of functions, defined over
  arbitrary free algebras have, when implemented as programs,
  polynomial memoized runtime complexity (see Theorem~\ref{t:sound}).
  By our invariance theorem, the function algebra is sound for
  polynomial time, and consequently GSRS characterizes the class of
  polytime computable functions.
\end{varenumerate}
\shortv{%
An extended version of this paper with more details, including all
proofs, is also available~\cite{AD:TR:14}.}

\subsection{Related Work}\label{sect:rw}
That predicative recursion \emph{on strings} is sound for polynomial
time, even in presence of simultaneous recursive definitions, is known
for a long time~\cite{Bellantoni:Diss:92}. Variations of predicative
recursion have been later considered and proved to characterize
classes like \PH~\cite{Bellantoni:FMII:94},
\PSPACE~\cite{Oitavem:PTCS:01},
\EXPTIME~\cite{AE:TOCL:09} or 
\NC~\cite{BKMO:LCS:08}. 
Predicative recursion on trees has been claimed to be sound for polynomial time
in the original paper by Leivant~\cite{Leivant:POPL:93}, the long
version of which only deals with
strings~\cite{Leivant:FMII:95}. After fifteen years, the
non-simultaneous case has been settled by the second author with
Martini and Zorzi~\cite{DMZ:DICE:10}; their proof, however, relies on
an ad-hoc, infinitary, notion of graph rewriting. Recently,
ramification has been studied in the context of a simply-typed
$\lambda$-calculus in an unpublished manuscript~\cite{DR:CoRR:12}; the
authors claim that a form of ramified recurrence on trees captures
polynomial time; this, again, does not take simultaneous recursion
into account.

The formalism presented here is partly inspired by the work of \citet{Hoffmann:ALP:92}, 
where sharing and memoization is shown to work well together in the realm
of term graph rewriting. 
The proposed machinery, although powerful, is unnecessarily
complicated for our purposes. Speaking in Hoffmann's terms, our results
require a form of full memoization, which \emph{is} definable 
in Hoffmann's system. However, most crucially for our concerns, 
it is unclear how the overall system incorporating full memoization
can be implemented efficiently, if at all.  


\section{The Need for Sharing and Memoisation}\label{s:examples}
This Section is an informal, example-driven, introduction to 
ramified recursive definitions and their complexity. Our objective is to
convince the reader that those definitions do \emph{not} give rise to
polynomial time computations if naively evaluated, and that sharing
and memoization are \emph{both} necessary to avoid exponential
blowups.

In Leivant's system~\cite{Leivant:FMII:95}, functions and variables
are equiped with a \emph{tier}.  Composition must preserve tiers and,
crucially, in a function defined by primitive recursion the tier of
the recurrence parameter must be higher than the tier of the recursive
call.  This form of \emph{ramification} of functions effectively tames
primitive recursion, resulting in a characterisation of the class of
\emph{polytime computable functions}.

Of course, ramification also controls the growth rate of
functions. However, as soon as we switch from strings to a domain
where tree structures are definable, this control is apparently lost.
For illustration, consider the following definition.
\begin{alignat*}{3}
  \fun{tree}(\czero) & = \cleaf & 
  \qquad\fun{tree}(\csuc(\var{n})) & = \fun{br}(\fun{tree}(\var{n})) & 
  \qquad\fun{br}(\var{t}) & = \cbranch(\var{t},\var{t}) \tpkt
\end{alignat*}

\begin{figure}%
  \begin{framed}
    \centering
    \hfill
    \subfloat[][Explicit tree representation.]{%
      {\small
        \begin{forest}
          [
          $\cbranch$, s sep=0mm
          [ $\cbranch$ 
          [$\cbranch$
          [$\cbranch$ [$\cleaf$] [$\cleaf$]] 
          [$\cbranch$ [$\cleaf$] [$\cleaf$]]] 
          [$\cbranch$ [$\cbranch$ [$\cleaf$] [$\cleaf$]] [$\cbranch$ [$\cleaf$] [$\cleaf$]]]]
          [ $\cbranch$ 
          [$\cbranch$ [$\cbranch$ [$\cleaf$] [$\cleaf$]] [$\cbranch$ [$\cleaf$] [$\cleaf$]]] 
          [$\cbranch$ [$\cbranch$ [$\cleaf$] [$\cleaf$]] [$\cbranch$ [$\cleaf$] [$\cleaf$]]]]
          ]
        \end{forest}
      }
    }%
    \hfill
    \subfloat[][Compact DAG.]{%
      {\small
        \begin{forest}
          [
          $\cbranch$, name=A, for tree={edge={draw=white}},
          [$\cbranch$, name=B,
          [$\cbranch$, name=C, 
          [$\cbranch$, name=D,
          [$\cleaf$, name=E]]]]]
          \draw (A) edge[bend left] (B);
          \draw (A) edge[bend right] (B);
          \draw (B) edge[bend left] (C);
          \draw (B) edge[bend right] (C);
          \draw (C) edge[bend left] (D);
          \draw (C) edge[bend right] (D);
          \draw (D) edge[bend left] (E);
          \draw (D) edge[bend right] (E);
          \node[xshift=-1.1cm] at (A) {};
          \node[xshift=1.1cm] at (A) {};
        \end{forest}
      }
    }%
    \hfill
  \end{framed}
  \caption{Complete Binary Tree of Height Four, as Computed by $\fun{tree}(\csuc^4(\czero))$.}%
  \label{fig:tree}%
\end{figure}
The function $\fun{tree}$ is defined by primitive recursion,
essentially from basic functions.  As a consequence, it is easily seen
to be ramified in the sense of Leivant.  Even though the the number of
recursive steps is linear in the input, the result of
$\fun{tree}(\csuc^n(\czero))$ is the complete binary tree of height
$n$.  As thus the length of the output is exponential in the one of
its input, there is, at least apparently, little hope to prove
$\fun{tree}$ a polytime function. The way out is \emph{sharing}: the
complete binary tree of height $n$ can be compactly represented as a
\emph{directed acyclic graph} (\emph{DAG} for short) of linear size
(see Figure~\ref{fig:tree}).  Indeed, using the compact DAG
representation it is easy to see that the function $\fun{tree}$ is
computable in polynomial time. This is the starting point
of~\cite{DMZ:DICE:10}, in which general ramified recurrence is proved
sound for polynomial time. A crucial observation here is that
\emph{not only} the output's size, but also the total amount of work
can be kept under control, thanks to the fact that evaluating a primitive
recursive definition on a compactly represented input can be done
by constructing an isomorphic DAG of recursive calls.

This does not scale up to \emph{simultaneous} ramified recurrence.
The following example computes the genealogical tree associated with
\emph{Fibonacci's rabbit problem} for $n\in\N$ generations. Rabbits
come in pairs. After one generation, each \emph{baby} rabbit pair
($\rabbitn$) matures. In each generation, an \emph{adult} rabbit pair
$(\rabbitm)$ bears one pair of babies.
\begin{alignat*}{5}
\frabbits(\czero) & = \rabbitln
&\qquad \fadults(\czero) & = \rabbitlm 
&\qquad \fbabies(\czero) & = \rabbitln\\
\frabbits(\csuc(\var{n})) & = \fbabies(\var{n})
&\qquad \fadults(\csuc(\var{n})) & = \rabbitm(\fadults(\var{n}),\fbabies(\var{n})) 
&\qquad \fbabies(\csuc(\var{n})) & = \rabbitn(\fadults(\var{n})) \tpkt
\end{alignat*}
The function $\frabbits$ is obtained by case analysis from the
functions $\fadults$ and $\fbabies$, which are defined by
\emph{simultaneous} primitive recursion: the former recursively calls
itself \emph{and} the latter, while the latter makes a recursive call
to the former. The output of $\frabbits(\csuc^n(\czero))$ is tightly
related to the sequence of Fibonacci numbers: the number of nodes at
depth $i$ is given by the \nth{$i$} Fibonacci number.  Hence the
output tree has exponential size in $n$ but, again, can be represented
compactly (see Figure~\ref{fig:fib}). This does not suffice for our
purposes, however. In presence of simultaneous definitions, indeed,
avoiding re-computation of previously computed values becomes more
difficult, the trick described above does not work, and the key idea
towards that is the use of \emph{memoization}.

\begin{figure}%
\centering
\begin{framed}
  \hfill
  \subfloat[][Explicit tree representation.]{%
    \small{
      \begin{forest}
        [$\rabbitn$,
        for tree={calign=first}
        [$\rabbitm$
        [$\rabbitm$ 
        [$\rabbitm$
        [$\rabbitm$ [$\rabbitlm$] [$\rabbitln$]] 
        [$\rabbitn$ [$\rabbitlm$]]] 
        [$\rabbitn$ 
        [$\rabbitm$ [$\rabbitlm$] 
        [$\rabbitln$]]]]
        [$\rabbitn$ 
        [$\rabbitm$
        [$\rabbitm$ [$\rabbitlm$] [$\rabbitln$]] 
        [$\rabbitn$ [$\rabbitlm$]]]]
        ]
        ]
      \end{forest}
    }
  }%
  \hfill
  \subfloat[][Compact DAG.]{%
    \small{
      \begin{forest}
        [$\rabbitn$
        [$\rabbitm$,
        for tree={calign=first}
        [$\rabbitm$,name=M3
        [$\rabbitm$,name=M2
        [$\rabbitm$,name=M1
        [$\rabbitlm$,name=M0]
        [$\rabbitln$,name=N0]] 
        [$\rabbitn$,name=N1]] 
        [$\rabbitn$,name=N2]]
        [$\rabbitn$,name=N3]
        ]]
        \draw (N3) -- (M2);
        \draw (N2) -- (M1);
        \draw (N1) -- (M0);
        \node[xshift=-8mm] at (M2) {};
        \node[xshift=1cm] at (N2) {};
      \end{forest}
    }
  }%
  \hfill
\end{framed}
\caption{Genealogical Rabbit Tree up to the Sixth Generation, as Computed by $\frabbits(\csuc^6(\czero))$.}%
\label{fig:fib}%
\end{figure}
What we prove in this paper is precisely that sharing and memoization can indeed
be made to work together, and that they together allow to prove polytime soundness
for all ramified recursive functions, also in presence of tree algebras and simultaneous
definitions.

\section{Preliminaries}\label{s:basics}

\paragraph*{General Ramified Simultaneous Recurrence}
Let $\faone$ denote a finite \emph{(untyped) signatures} $\FS$ of \emph{constructors} $\seq[k]{\conone}$, 
each equipped with an arity $\arity{\conone_i}$. 
In the following, the set of terms $\TERMS[\faone]$ is also denoted by $\faone$ if this does not create ambiguities. 
We are interested in total functions from
$\faone^n=\underbrace{\faone\times\ldots\times\faone}_{\mbox{$n$ times}}$
to $\faone$. 

\begin{definition}\label{d:basicfs}
The following are so-called \emph{basic functions}:
\begin{varitemize}
\item
  For each constructor $\conone$, the \emph{constructor function} 
  $\funone_{\conone}:\faone^{\arity{\conone}}\rightarrow\faone$ for $\conone$,
  defined as follows:
  $\funone_{\conone}(\varone_1,\ldots,\varone_{\arity{\conone}})=\conone(\varone_1,\ldots,\varone_{\arity{\conone}})$
\item
  For each $1\leq n\leq m$, the $(m,n)$-\emph{projection function}
  $\proj{m}{n}:\faone^m\rightarrow\faone$ defined as follows:
  $\proj{m}{n}(\varone_1\ldots,\varone_m)=\varone_n$.
\end{varitemize}
\end{definition}

\begin{definition}\label{d:schemes}
\envskipline%
\begin{varitemize}
\item
  Given a function $\funone:\faone^n\rightarrow\faone$ and $n$ functions
  $\funtwo_1,\ldots,\funtwo_n$, all of them from $\faone^m$ to $\faone$,
  the \emph{composition} $\funthree=\comp{\funone}{(\funtwo_1,\ldots,\funtwo_n)}$
  is a function from $\faone^m$ to $\faone$ defined as follows:
  $\funthree(\vec{\varone})=\funone(\funtwo_1(\vec{\varone}),\ldots,\funtwo_n(\vec{\varone}))$.
\item
  Suppose given the functions $\funone_i$ where $1\leq i\leq k$ such that for some $m$,
  $\ofdom{\funone_i}{\faone^{\arity{\conone_i}}\times\faone^{n}\rightarrow\faone}$.
  Then the function $\funtwo=\casdist{\{\funone_i\}_{1\leq i\leq k}}$
  defined by \emph{case distinction} from $\{\funone_i\}_{1\leq i\leq k}$
  is a function from $\faone\times\faone^n$ to $\faone$ defined as follows:
  $\funtwo(\conone_i(\vec{\varone}),\vec{\vartwo})=\funone_i(\vec{\varone},\vec{\vartwo})$.
\item
  Suppose given the functions $\funone_i^j$, where $1\leq i\leq k$ and $1\leq j\leq n$, such that
  for some $m$,
  $\ofdom{\funone_i^j}{\faone^{\arity{\conone_i}}\times\faone^{n\cdot\arity{\conone_i}}\times\faone^{m}\rightarrow\faone}$.
  The functions $\{\funtwo_j\}_{1\leq j\leq n}=\simrec{\{\funone_i^j\}_{1\leq i\leq k,1\leq j\leq n}}$ defined
  by \emph{simultaneous primitive recursion} from $\{\funone_i^j\}_{1\leq i\leq k,1\leq j\leq n}$
  are all functions from $\faone\times\faone^m$ to $\faone$ such that
  for $\vec{\varone} = \seq[\arity{{\conone}_i}]{\varone}$, 
  \[
  \funtwo_j(\conone_i(\vec{\varone}),\vec{\vartwo})=
  \funone_i^j(\vec{\varone},\funtwo_1(\varone_1,\vec{\vartwo}),\ldots,\funtwo_1(\varone_{\arity{\conone_i}},\vec{\vartwo}),\ldots,
      \funtwo_n(\varone_1,\vec{\vartwo}),\ldots,\funtwo_n(\varone_{\arity{\conone_i}},\vec{\vartwo}),\vec{\vartwo})
      \tpkt
  \]
\end{varitemize}
\end{definition}
We denote by $\SIMREC[\faone]$ the class of \emph{simultaneous recursive functions over $\faone$}, 
defined as the smallest class containing the basic functions of 
Definition~\ref{d:basicfs} and that is closed under the schemes of Definition~\ref{d:schemes}.

\emph{Tiering}, the central notion underlying Leivant's definition of
\emph{ramified recurrence}, consists in attributing \emph{tiers} to
inputs and outputs of some functions among the ones constructed as
above, with the goal of isolating the polytime computable ones.
Roughly speaking, the role of tiers is to single out ``a copy'' of
the signature by a level: this level permits to control the recursion nesting.
Tiering can be given as a formal system, in which
judgments have the form
$\tj{\funone}{\faone_{p_1}\times\ldots\times\faone_{p_{\arity{\funone}}}}{\faone_m}$
for $\seq[\arity{\funone}]{p},m$ natural numbers and $\funone \in \SIMREC[\faone]$.
The system is defined in Figure~\ref{fig:typing}, where 
$\multyone$ denotes the expression $\faone_{q_1}\times\ldots\times\faone_{q_k}$
for some $\seq[k]{q} \in \N$. 
Notice that composition preserves tiers. 
Moreover, recursion is allowed only on
inputs of tier higher than the tier of the function (in the case
$\funone = \simrec{\{\funone_i^j\}_{1\leq i\leq k,1\leq j\leq n}}$, we
require $p > m$).
\begin{definition}
  We call a function $\funone \in \SIMREC[\faone]$ definable by \emph{general ramified simultaneous recurrence} (\emph{GRSR} for short) if
  $\tj{\funone}{\faone_{p_1}\times\ldots\times\faone_{p_{\arity{\funone}}}}{\faone_m}$
  holds.
\end{definition}
\begin{remark}
  Consider the \emph{word algebra} $\faword = \{\cempty,\cwordzero,\cwordone\}$ consisting 
  of a constant $\cempty$ and two unary constructors $\cwordzero$ and $\cwordone$, 
  which is in bijective correspondence to the set of binary words. 
  Then the functions definable by ramified simultaneous recurrence over $\faword$ includes 
  the ramified recursive functions from \citet{Leivant:FMII:95}, and consequently all polytime computable functions.
\end{remark}

\newcommand{\TOP}{\rule{0pt}{2.5ex}}
\begin{figure}
  \centering
  \begin{framed}
    \[
    \infer
    {\TOP\tj{\funone_{\conone}}{\faone_{n}^{\arity{\conone}}}{\faone_n}}
    {}
    \hfill
    \infer
    {\TOP\tj{\proj{n}{m}}{\faone_{p_1}\times\ldots\times\faone_{p_m}}{\faone_{p_n}}}
    {}
    \hfill
    \infer
    {\TOP\tj{\casdist{\{\funone_i\}_{1\leq i\leq k}}}{\faone_p\times\multyone}{\faone_m}}
    {
      \tj{\funone_i}{\faone^{\arity{\conone_i}}_p\times\multyone}{\faone_m}
    }
    \]
    \[
    \infer
    {\strut\tj{\comp{\funone}{(\funtwo_1,\ldots,\funtwo_n)}}{\multyone}{\faone_m}}
    {
      \tj{\funone}{\faone_{p_1}\times\ldots\times\faone_{p_n}}{\faone_m} & \tj{\funtwo_i}{\multyone}{\faone_{p_i}}
    }
    \hfill
    \infer
    {\strut\tj{\simrec{\{\funone_i^j\}_{1\leq i\leq k,1\leq j\leq n}}}{\faone_p\times\multyone}{\faone_m}}
    {
      \tj{\funone_i^j}{\faone^{\arity{\conone_i}}_p\times \faone^{n \cdot \arity{\conone_i}}_m\times\multyone}{\faone_m}  
      & p > m
    }
    \]
  \end{framed}
\caption{Tiering as a Formal System.}
\label{fig:typing}
\end{figure}
\begin{example}\label{ex:algebra}
  \envskipline
  \begin{varenumerate}
  \item Consider $\fanat \defsym \{ \czero, \csuc \}$ with $\arity{\czero} = 0$ and $\arity{\csuc} = 1$,
    which is in bijective correspondence to the set of natural numbers. 
    We can define addition $\ofdom{\fun{add}}{\fanat_i \times \fanat_j \rightarrow \fanat_j}$ 
    for $i > j$, by
    \begin{align*}
      \fun{add}(\czero,\var{y}) & = \proj{1}{1}(\var{y}) = \var{y} 
      & \fun{add}(\csuc(\var{x}),\var{y}) & = (\comp{\funone_{\csuc}}{\proj{3}{2}})(\var{x},\fun{add}(\var{x},\var{y}),\var{y}) = \csuc(\fun{add}(\var{x},\var{y}))
      \tkom
    \end{align*}
    using general simultaneous ramified recursion, i.e. $\{\fun{add}\} = \simrec{\{\{\proj{1}{1},\comp{\funone_{\csuc}}{\proj{3}{2}}\}\}}$.
  \item Let $\farabbits \defsym \{\rabbitln, \rabbitlm, \rabbitn, \rabbitm\}$, where
    $\arity{\rabbitln} = \arity{\rabbitlm} = 0$, $\arity{\rabbitn} = 1$ and $\arity{\rabbitm} = 2$.
    Then we can define the functions $\ofdom{\frabbits}{\fanat_i \rightarrow \farabbits_j}$ for $i > j$ 
    from Section~\ref{s:examples} by composition from the following two functions, 
    defined by simultaneous ramified recurrence.
    \begin{align*}
      \fadults(\czero) & = \rabbitlm
      & \fadults(\csuc(\var{n})) & = (\comp{\funone_{\rabbitm}}{(\proj{3}{2},\proj{3}{3})})\,(\var{n},\fadults(\var{n}),\fbabies(\var{n})) = \rabbitm(\fadults(\var{n}),\fbabies(\var{n}))\\
      \fbabies(\czero) & = \rabbitln 
      & \fbabies(\csuc(\var{n})) & = (\comp{\funone_{\rabbitn}}{\proj{3}{3}})\,(\var{n},\fadults(\var{n}),\fbabies(\var{n})) = \rabbitn(\fadults(\var{n}))
      \tpkt
    \end{align*}
  \item 
    We can define a function
    $\ofdom{\fun{\#leafs}}{\farabbits \rightarrow \fanat}$
    by simultaneous primitive recursion
    which counts the number of leafs in $\farabbits$-trees as follows.
    \begin{align*}
      \fun{\#leafs}(\rabbitln) & = \csuc(\czero) 
      & \fun{\#leafs}(\rabbitlm) & = \csuc(\czero) \\
      \fun{\#leafs}(\rabbitn(\var{t})) & = \fun{\#leafs}(\var{t})
      & \fun{\#leafs}(\rabbitm(\var{l},\var{r})) & = \fun{add}(\fun{\#leafs}(\var{l}),\fun{\#leafs}(\var{r}))
      \tpkt
    \end{align*}
    However, this function \emph{cannot} be ramified, since $\fun{add}$ in the last 
    equation requires different tiers. 
    Indeed, having a ramified recursive function 
    $\ofdom{\fun{\#leafs}}{\farabbits_i \rightarrow \fanat_1}$ (for some $i > 1$) 
    defined as above would allow us to ramify $\fun{fib} = \comp{\fun{\#leafs}}{\frabbits}$ which 
    on input $n$ computes the \nth{$n$} Fibonacci number, and is thus an exponential function. 
  \end{varenumerate}
\end{example}

\paragraph*{Computational Model, Syntax and Semantics}
We introduce a simple, \emph{rewriting based}, notion of program for computing
functions over term algebras.
\begin{definition}
  A \emph{program} $\progone$ is given as a triple
  $\prog{\DS}{\CS}{\TRSone}$ consisting of two disjoint signatures 
  $\DS$ and $\CS$ of \emph{operation symbols} $\seq[m]{\symone}$ and 
  \emph{constructors} $\seq[n]{\conone}$ respectively,
  and a \emph{finite} set $\TRSone$ of \emph{rules} 
  $l \to r$ over terms $l,r \in \TERMS[\DS \cup \CS][\VS]$. 
  For each rule, the \emph{left-hand side} $l$ is of the form 
  $\symone_i(\seq[k]{p})$ where the \emph{patterns} $p_j$ consist only of variables and constructors, 
  and all variables occurring in the \emph{right-hand side} $r$ also occur in the left-hand side $l$. 
\end{definition}
We keep the program $\progone = \prog{\DS}{\CS}{\TRSone}$ fixed throughout the following. 
Moreover, we require that $\progone$ is \emph{orthogonal}, that is, the following two requirements are met:
\begin{varenumerate}
\item\emph{left-linearity:} the left-hand sides $l$ of each rule $l \to r \in \TRSone$ is \emph{linear}; and
\item\emph{non-ambiguity:} there are no two rules with overlapping left-hand sides in $\TRSone$.
\end{varenumerate}
Orthogonal programs define a class of deterministic first-order functional programs, see e.g.\ \cite{BN:1998}.
The domain of the defined functions is the constructor algebra $\TERMS[\CS]$.
Correspondingly, elements of $\TERMS[\CS]$ are called \emph{values}, 
which we denote by $\valone,\valtwo,\dots$\,.
\begin{figure}
  \centering
    \begin{framed}
      \[
      ~
      \]
      \[
      \infer
      {\symone(\seq[k]{\termone}) \reduce \valone}
      { \funone \in \DS
        & \termone_i \reduce \valone_i
        & \symone(\seq[k]{\valone}) \reduce \valone
      }
      \hfill
      \infer
      {\conone(\seq[k]{\termone}) \reduce \conone(\seq[k]{\valone})}
      {\conone \in \CS & \termone_i \reduce \valone_i}
     \hfill
      \hfill
      \]
      \[
      \hfill
      \infer
      {\symone(\seq[k]{\valone}) \reduce \valone}
      { \symone(\seq[k]{p}) \to r \in \TRSone 
        & \forall i.\ p_i\sigma = \valone_i
        & r\sigma \reduce \valone
      }
      \hfill
      \]
  \end{framed}
  \caption{Operational Semantics for Program $\prog{\DS}{\CS}{\TRSone}$.}
  \label{fig:os}
\end{figure}
In Figure~\ref{fig:os} we present the operational semantics, 
realizing standard \emph{call-by-value} evaluation order.  
The statement $\termone \reduce \valone$ means that the term $\termone$ \emph{reduces} to the value $\valone$.
We say that $\progone$ computes the function $\ofdom{\funone}{\TERMS[\CS]^k \rightarrow \TERMS[\CS]}$
if there exists an operation $\symone \in \DS$, such that 
$\funone(\seq[k]{\valone}) = \valone$ if and only if $\symone(\seq[k]{\valone}) \reduce \valone$
holds for all inputs $\valone_i \in \TERMS[\CS]$. 

\begin{example}[Continued from Example~\ref{ex:algebra}]\label{ex:program}
  The definition of $\frabbits$ from Section~\ref{s:examples} can be turned into a program $\prograbbits$ 
  over constructors of $\fanat$ and $\farabbits$, 
  by \emph{orienting} the underlying equations from left to right and replacing 
  applications of functions $\funone \in \{{\frabbits},{\fadults},{\fbabies}\}$ with corresponding
  operation symbols $\symone \in \{{\srabbits},{\sadults},{\sbabies}\}$. For instance, 
  concerning the function $\fadults$, the defining equations are turned into 
  $\sadults(\czero) \to \rabbitlm$ and $\sadults(\csuc(\var{n})) \to \rabbitm(\sadults(\var{n}),\sbabies(\var{n}))$.
\end{example}
\shortv{%
The example hints at a systematic construction of programs $\progone_{\funone}$ computing 
functions $\funone \in \SIMREC[\faone]$, which can be made precise~\cite{AD:TR:14}.
}
\longv{
\begin{definition}
  For $\funone \in \SIMREC[\faone]$, 
  by $\progone_{\funone}$ we denote the program $\prog{\DS_{\funone}}{\CS_{\funone}}{\TRSone_{\funone}}$ where:
  \begin{varitemize}
  \item the set of operations $\DS_{\funone}$ contains for each function $\funtwo$ 
    underlying the definition of $\funone$ a corresponding operation symbol $\symtwo$.
  \item the set of constructors $\CS_{\funone}$ contains the constructors of $\faone$; 
  \item the set of rules $\TRSone_{\funone}$ contain for each equation 
    $l = r$ defining a function  $\funtwo$ underlying the definition $\funone$
    the \emph{orientation} $l \to r$. 
  \end{varitemize}
\end{definition}
Notice that due to the inductive definition of the class $\SIMREC[\faone]$, 
the program $\progone_{\funone}$ is finite. From the shape of the initial functions 
and operations (Definition~\ref{d:basicfs} and Definition~\ref{d:schemes}) it 
is also clear that $\progone_{\funone}$ is orthogonal.
}
\paragraph*{Terms and Term Graphs}
Furthermore, we fix a set of \emph{variables} $\VS$ disjoint from function symbols. 
Terms over a signature $\FS$ and $\VS$ are defined as usual, 
and form a set $\TERMS[\FS][\VS]$. 
A term $\termone$ is called \emph{ground} if it does not contain variables, 
it is called \emph{linear} if every variable occurs at most once in $\termone$. 
The ground terms are collected in $\TERMS[\FS]$.
The set of \emph{subterms} $\subterms{\termone}$ of a term $\termone$ is defined by 
$\subterms{\termone} \defsym \{\termone\}$ if $\termone \in \VS$ and 
$\subterms{\termone} \defsym \bigcup_{1\leq i\leq\arity{f}}\subterms{\termone_i} \cup\{\termone\}$ 
if $\termone = f(\termone_1,\ldots,\termone_{\arity{f}})$. 
A \emph{substitution}, is a finite mapping $\sigma$ from variables to terms. 
By $\termone\sigma$ we denote the term obtained by replacing in $\termone$ 
all variables $\varone \in \dom(\sigma)$ by $\sigma(\varone)$. 
If $\termtwo = \termone\sigma$, we also say that $\termtwo$ is an \emph{instance} of the term $\termone$. 

We borrow key concepts from \emph{term graph rewriting} (see e.g. 
the survey of Plump~\cite{Plump:ENTCS:01} for an overview) and 
follow the presentation of \citet{BEGKPS:PARLE:87}.  
A \emph{term graph} $\tgraphone$ over a signature $\FS$ is 
a \emph{directed acyclic graph} whose nodes are labeled by symbols in $\FS \cup \VS$, 
and where outgoing edges are ordered. 
Formally, $\tgraphone$ is a triple $(\nodes,\suc,\lab)$ consisting of \emph{nodes} $\nodes$, 
a \emph{successors function} $\ofdom{\suc}{\nodes \to \nodes^{\ast}}$ and a 
\emph{labeling function} $\ofdom{\lab}{\nodes \to \mathcal{\FS \cup \VS}}$. 
We require that term graphs are \emph{compatible} with $\FS$, in the sense that 
for each node $\nodeone \in \nodes$,
if $\lab[\tgraphone](\nodeone) = f \in \FS$ then $\suc[\tgraphone](\nodeone) = \lseq[\arity{f}]{\nodeone}$ 
and otherwise, if $\lab[\tgraphone](\nodeone) = \varone \in \VS$, $\suc[\tgraphone](\nodeone) = \nil$. 
In the former case, we also write $\tgraphone(\nodeone) = f(\seq[\arity{f}]{\nodeone})$, the latter case
is denoted by $\tgraphone(\nodeone) = \varone$. 
We define the \emph{successor relation} $\reach[\tgraphone]$ on nodes in $\tgraphone$ such that
$\nodeone \reach[\tgraphone] \nodetwo$ holds iff $\nodetwo$ occurs in $\suc(\nodeone)$, 
if $\nodetwo$ occurs at the \nth{$i$} position we also write $\nodeone \suci[\tgraphone]{i} \nodetwo$. 
Throughout the following, we consider only \emph{acyclic} term graphs, 
that is, when $\reach[\tgraphone]$ is acyclic. 
Hence the \emph{unfolding} $\trepr[\tgraphone]{\nodeone}$ of $\tgraphone$ \emph{at} node $\nodeone$,
defined by $\trepr[\tgraphone]{\nodeone} \defsym \varone$ if $\tgraphone(\nodeone) = \varone \in \VS$, 
and otherwise $\trepr[\tgraphone]{\nodeone} \defsym f(\trepr[\tgraphone]{\nodeone_1},\dots,\trepr[\tgraphone]{\nodeone_k})$ 
where $\tgraphone(\nodeone) = f(\seq[k]{\nodeone})$,
results in a finite term. 
We called the term graph $\tgraphone$ \emph{rooted} if there
exists a unique node $\nodeone$, the \emph{root} of
$\tgraphone$, with $\nodeone \reachtr[\tgraphone] \nodetwo$ for every
$\nodetwo \in \nodes$.  We denote by
$\subgraphAt{\tgraphone}{\nodeone}$ the \emph{sub-graph} of $\tgraphone$
\emph{rooted} at $\nodeone$.
Consider a symbol $f \in \FS$ and nodes $\sseq[\arity{f}]{\nodeone} \subseteq \nodes$
of $\tgraphone$.
The extension $\tgraphtwo$ of $\tgraphone$ by a fresh node
$\nodeone_f \not \in \nodes$ with $\tgraphtwo(\nodeone_f) =
f(\seq[\arity{f}]{\nodeone})$ is denoted by 
$\ginsert{\tgraphone}{\nodeone_f}{f}{\seq[\arity{f}]{\nodeone}}$.
We write $f(\subgraphAt{\tgraphone}{\nodeone_1}, \dots, \subgraphAt{\tgraphone}{\nodeone_{\arity{f}}})$ 
for the term graph $\subgraphAt{\tgraphtwo}{\nodeone_f}$.

For two rooted term graphs $\tgraphone = (\nodes[\tgraphone],\suc[\tgraphone],\lab[\tgraphone])$ and 
$\tgraphtwo = (\nodes[\tgraphtwo],\suc[\tgraphtwo],\lab[\tgraphtwo])$, 
a mapping $\ofdom{\morphone}{\nodes[\tgraphone] \to \nodes[\tgraphtwo]}$ is called 
\emph{morphic} in $\nodeone \in \nodes[\tgraphone]$ if
(i)~$\lab[\tgraphone](\nodeone) = \lab[\tgraphtwo](\morphone(\nodeone))$ and
(ii)~$\nodeone \suci[\tgraphone]{i} \nodetwo$ implies $\morphone(\nodeone) \suci[\tgraphtwo]{i} \morphone(\nodetwo)$ for all appropriate $i$.
A \emph{homomorphism}
from $\tgraphone$ to $\tgraphtwo$ is a mapping
$\ofdom{\morphone}{\nodes[\tgraphone] \to \nodes[\tgraphtwo]}$ that 
(i)~maps the root of $\tgraphone$ to the root of $\tgraphtwo$ and that 
(ii)~is morphic in all nodes $\nodeone \in \nodes[\tgraphone]$ not labeled by a variable.
We write $\mmatch{\morphone}{\tgraphone}{\tgraphtwo}$ to indicate that $\morphone$ 
is, possibly an extension of, a homomorphism from $\tgraphone$ to $\tgraphtwo$.

Every term $\termone$ is trivially representable as a \emph{canonical tree} $\tree{\termone}$ unfolding to $\termone$, 
using a fresh node for each occurrence of a subterm in $\termone$. 
For $\termone$ a linear term, to each variable $\varone$ in $\termone$ we can 
associate a \emph{unique node} in $\tree{\termone}$ labeled by $\varone$, 
which we denote by $\nodeone_\varone$. 
The following proposition relates matching on terms and homomorphisms on trees.
It essentially relies on the imposed linearity condition. 
\begin{proposition}[Matching on Graphs]\label{p:tg:match}
  Let $\termone$ be a linear term, $\tgraphone$ be a term graph and let $\nodeone$ be a node of $\tgraphone$.
  \begin{enumerate}
  \item If $\mmatch{\morphone}{\tree{\termone}}{\subgraphAt{\tgraphone}{\nodeone}}$ then 
    there exists a substitution $\sigma$ such that $\termone\sigma = \trepr[\tgraphone]{\nodeone}$.
  \item Vice versa, if $\termone\sigma = \trepr[\tgraphone]{\nodeone}$ holds for some substitution $\sigma$ 
    then there exists a morphism $\mmatch{\morphone}{\tree{\termone}}{\subgraphAt{\tgraphone}{\nodeone}}$.
  \end{enumerate}
  Here, the substitution $\sigma$ and homomorphism $\morphone$ satisfy
  $\sigma(\varone) = \trepr[\tgraphone]{\morphone(\nodeone_\varone)}$ for all variables $\varone$ in $\termone$.
\end{proposition}
\longv{
\begin{proof}
  The proof is by induction on $t$. 
  We first proof the direction from left to right. 
  Assume $\mmatch{\morphone}{\tree{\termone}}{\subgraphAt{\tgraphtwo}{\nodeone}}$.
  When $\termone$ is a variable, the substitution $\sigma \defsym \{ \termone \mapsto \trepr[\tgraphtwo]{\nodeone} \}$
  satisfies $\termone\sigma = \trepr[\tgraphtwo]{\nodeone}$. 
  Since $\morphone(\posempty) = \nodeone$, we conclude the base case. 
  For the inductive step, assume $\termone = f(\seq[k]{\termone})$. 
  Fix $i = 1,\dots,k$, and define $\morphone_i(p) = \morphone(i \posc p)$ for each position 
  $i \posc p$ in $t$. By case analysis on the nodes of $\tree{\termone_i}$
  one verifies $\mmatch{\morphone_i}{\tree{\termone_i}}{\trepr[\tgraphtwo]{n_i}}$.
  Thus by induction hypothesis $\termone_i\sigma_i = \trepr[\tgraphtwo]{n_i}$ for a substitution $\sigma_i$,
  where without loss of generality $\sigma_i$ is restricted to variables in $\termone_i$. 
  Define $\sigma \defsym \bigcup_{i=1}^k\sigma_i$. Then
  $\termone\sigma = f(\termone_1\sigma_1,\dots,\termone_k\sigma_k) 
  = f(\trepr[\tgraphtwo]{\nodeone_1},\dots,\trepr[\tgraphtwo]{\nodeone_k}) = \trepr[\tgraphtwo]{\nodeone}$, 
  where the last equality follows as $\morphone$ is morphic on $\nodeone$.
  Moreover, from the shape of $\sigma_i$ and $\morphone_i$ it is not difficult to see 
  that by construction the substitution $\sigma$ and homomorphism $\morphone$ are related 
  as claimed by the lemma. 
  
  Now for the inverse direction, suppose $\termone\sigma = \trepr[\tgraphtwo]{\nodeone}$. 
  If $\termone$ is a variable and thus $\tree{\termone}$ consists of a single unlabeled node,
  trivially $\mmatch{\morphone}{\tree{\termone}}{\subgraphAt{\tgraphtwo}{\nodeone}}$ holds
  for $\morphone$ the homomorphism which maps the root of $\tree{\termone}$ to $\nodeone$. 
  Observe that $\sigma(\termone) = \trepr[\tgraphtwo]{\nodeone} = \trepr[\tgraphtwo]{\morphone(\posempty)}$, 
  which concludes the base case. 
  For the inductive step suppose $\termone = f(\seq[k]{\termone})$, 
  hence $\tgraphtwo(\nodeone) = f(\seq[k]{\nodeone})$, $\termone_i\sigma = \trepr[\tgraphtwo]{\nodeone_i}$ ($i = 1,\dots,k$)
  and thus by induction hypothesis $\mmatch{\morphone_i}{\tree{\termone_i}}{\trepr[\tgraphtwo]{\nodeone_i}}$ for homomorphisms $\morphone_i$.
  Define the function $\morphone$ by $\morphone(\posempty) \defsym \nodeone$ and $\morphone(i \posc p) \defsym \morphone_i(p)$ for 
  all $i = 1,\dots,k$ and positions $i \posc p$ of $\termone$.
  Observe that $\morphone$ is defined on all nodes of $\tree{\termone}$. 
  By definition of $\morphone$ one finally concludes the lemma, using 
  the induction hypotheses together with the equalities $\tree{\termone}(i \posc p) = \tree{\termone_i}(p)$ 
  for nodes $i \posc p$ $(i = 1, \dots, k)$ of $\tree{\termone}$.
\end{proof}
}

\section{Memoization and Sharing, Formally}\label{s:invariance}
To incorporate \emph{memoization}, we make
use of a \emph{cache} $\cacheone$ which stores results of intermediate
functions calls.  A \emph{cache} $\cacheone$ is modeled as a set of
tuples $(\symone(\seq[\arity{\symone}]{\valone}),\valone)$, where
$\symone \in \DS$ and $\seq[\arity{\symone}]{\valone}$ as well as $\valone$ are
values.
\begin{figure}
  \begin{framed}
    \centering
    \[
    \hfill
    \infer[Split]
    {\ctpair{\cacheone_0}{\symone(\seq[k]{\termone})} \reducem[m] \ctpair{\cacheone_{k+1}}{\valone}}
    { \symone \in \DS
      &\!\ctpair{\cacheone_{i-1}}{\termone_i} \reducem[n_i] \ctpair{\cacheone_i}{\valone_i} 
      &\!\ctpair{\cacheone_k}{\symone(\seq[k]{\valone})} \reducem[n] \ctpair{\cacheone_{k+1}}{\valone} 
      &\!\! m = n + \sum_{i=1}^k n_i
    }
    \hfill
    \]
    \[
    \hfill
    \infer[Con]
    {\ctpair{\cacheone_0}{\conone(\seq[k]{\termone})} \reducem[m] \ctpair{\cacheone_k}{\conone(\seq[k]{\valone})}}
    {\conone \in \CS & \ctpair{\cacheone_{i-1}}{\termone_i} \reducem[n_i] \ctpair{\cacheone_i}{\valone_i} & m = \sum_{i=1}^k n_i}
    \qquad
    \infer[Read]
    {\ctpair{\cacheone}{\symone(\seq[k]{\valone})} \reducem[0] \ctpair{\cacheone}{\valone}}
    {(\symone(\seq[k]{\valone}),\valone) \in \cacheone}
    \hfill
    \]
    \[
    \hfill
    \infer[Update]
    {\ctpair{\cacheone}{\symone(\seq[k]{\valone})} \reducem[m+1] \ctpair{\cachetwo \cup \{(\symone(\seq[k]{\valone}),\valone)\}}{\valone}}
    { (\symone(\seq[k]{\valone}),\valone) \not\in \cacheone 
      & \symone(\seq[k]{p}) \to r \in \TRSone 
      & \forall i.\ p_i\sigma = \valone_i
      & \ctpair{\cacheone}{r\sigma} \reducem[m] \ctpair{\cachetwo}{\valone}
    }
    \hfill
    \]
  \end{framed}
  \caption{Cost Annotated Operational Semantics with Memoization for Program $\prog{\DS}{\CS}{\TRSone}$.}
  \label{fig:os:memo}
\end{figure}
Figure~\ref{fig:os:memo} collects the \emph{memoizing operational
  semantics} with respect to the program $\progone =
\prog{\DS}{\CS}{\TRSone}$.  Here, a statement
$\ctpair{\cacheone}{\termone} \reducem[m] \ctpair{\cachetwo}{\valone}$
means that starting with a cache $\cacheone$, the term $\termone$
\emph{reduces} to the value $\valone$ with updated cache $\cachetwo$.
The natural number $m$ indicating the \emph{cost} of this reduction.
The definition is along the lines of the standard semantics
(Figure~\ref{fig:os}), carrying the cache throughout the reduction of
the given term.  The last rule of Figure~\ref{fig:os} is split into
two rules \rlname{Read} and \rlname{Update}. The former performs a
read from the cache, the latter the reduction in case the
corresponding function call is not tabulated, updating the cache with
the computed result.  Notice that in the semantics, a read is
attributed zero cost, whereas an update is accounted with a cost of
one.  Consequently the cost $m$ in $\ctpair{\cacheone}{\termone}
\reducem[m] \ctpair{\cachetwo}{\valone}$ refers to the number of
non-tabulated function applications.

\begin{lemma}
  We have $\ctpair{\cacheempty}{\termone} \reducem[m] \ctpair{\cacheone}{\valone}$ for some $m \in \N$ and cache $\cacheone$ if and only if
  $\termone \reduce \valone$.
\end{lemma}
\longv{
\begin{proof}
  Call a cache $\cacheone$ \emph{proper} if 
  $(\symone(\seq[k]{\valone}),\valone) \in \cacheone$ implies $\symone(\seq[k]{\valone}) \reduce \valone$. 
  For the direction from right to left, we show the following stronger claim.
  \begin{claim}
    Suppose $\termone \reduce \valone$ and let cache $\cacheone_1$ be a proper cache.
    Then $\ctpair{\cacheone_1}{\termone} \reducem[m] \ctpair{\cacheone_2}{\valone}$ for some $m$. 
  \end{claim}
  The proof is by induction on the deduction $\Pi$ of the statement $\termone \reduce \valone$.
  \begin{enumerate}
  \item Suppose that the last rule in $\Pi$ has the form
    \[
    \infer
    {\symone(\seq[k]{\termone}) \reduce \valone}
    { \termone_i \reduce \valone_i
      & \symone(\seq[k]{p}) \to r \in \TRSone 
      & p_i\sigma = \valone_i
      & r\sigma \reduce \valone
    }
    \]
    We consider the more involved case where at least one $\termone_i$ is not a value. 
    By induction hypothesis, we obtain proper caches $\seq[0][k]{\cachetwo}$ with 
    $\cachetwo_0 = \cacheone_1$ and $\ctpair{\cachetwo_{i-1}}{\termone_i} \reducem[m_i] \ctpair{\cachetwo_{i-1}}{\valone_i}$. 
    By the rule \rlname{Split}, it suffices to show 
    $\ctpair{\cachetwo_{k}}{\symone(\seq[k]{\valone})} \reducem[n] \ctpair{\cacheone_2}{\valone}$
    for $\cacheone_2$ a proper cache. 
    We distinguish two cases. 
    Consider the case $(\symone(\seq[k]{\valone}),\valtwo) \in \cachetwo_{k}$ for some $\valtwo$. 
    Using that $\symone(\seq[k]{\valone}) \reduce \valtwo$ implies $\valone = \valtwo$ for orthogonal programs, 
    we conclude the case by one application of rule \rlname{Read}. 
    Otherwise, we conclude by rule \rlname{Update} using the induction hypothesis on $r\sigma \reduce \valone$. 
    Note that the resulting cache is also in this case proper. 
  \item The final case follows directly from induction hypothesis, 
    using the rule $\rlname{Constructor}$. 
  \end{enumerate}

  For the direction from left to right we show the following stronger claim
  \begin{claim}
    Suppose $\ctpair{\cacheone_1}{\termone} \reducem[m] \ctpair{\cacheone_2}{\valone}$ for a proper cache $\cacheone_1$.
    Then $\termone \reduce \valone$. 
  \end{claim}
  The proof is by induction on the deduction $\Pi$ of the statement
  $\ctpair{\cacheone_1}{\termone} \reducem[m] \ctpair{\cacheone_2}{\valone}$
  \begin{enumerate}
  \item Suppose first that the last rule in $\Pi$ is of the form 
    \[
    \infer
    {\ctpair{\cacheone_0}{\symone(\seq[k]{\termone})} \reducem[m] \ctpair{\cacheone_{k+1}}{\valone}}
    { \ctpair{\cacheone_{i-1}}{\termone_i} \reducem[m_i] \ctpair{\cacheone_i}{\valone_i} 
      & \ctpair{\cacheone_k}{\symone(\seq[k]{\valone})} \reducem[n] \ctpair{\cacheone_{k+1}}{\valone} 
    }
    \]
    By induction hypothesis, we see $\termone_i \reduce \valone_i$, using also 
    that the configurations $\confone_i$ are all proper by the previous claim. 
    As we also have $\symone(\seq[k]{\valone}) \reduce \valone$ by induction 
    hypothesis, it follows that some rule $\symone(\seq[k]{p}) \to r \in \TRSone$ 
    matches $\symone(\seq[k]{\valone})$. Putting things together,
    we conclude by one application $\rlname{Function}$.
  \item The remaining cases where the last rule in $\Pi$ is $\rlname{Constructor}$, $\rlname{Read}$ or $\rlname{Update}$
    follow either from the assumption that $\cacheone_1$ is proper, 
    or from induction hypothesis using that 
    by the previous claim the intermediate and resulting caches are all proper. 
  \end{enumerate}
\end{proof}}
The lemma confirms that the call-by-value semantics of Section~\ref{s:basics} 
is correctly implemented by the memoizing semantics. 
To tame the growth rate of values, 
we define \emph{small-step semantics} corresponding to the memoizing semantics, 
facilitating sharing of common sub-expressions. 

\paragraph*{Small-Step Semantics with Memoization and Sharing}
To incorporate sharing, we extend the pairs 
$\ctpair{\cacheone}{\termone}$ by a \emph{heap}, 
and allow \emph{references} to the heap both in terms and in caches. 
Let $\Loc$ denote a countably infinite set of \emph{locations}.
We overload the notion of \emph{value}, and define \emph{expressions} $\expone$ 
and \emph{(evaluation) contexts} $\evalctxone$ according to the following grammar:
\begin{alignat*}{3}
\valone
& \defsym \,\refone
&& \mid \conone(\seq[k]{\valone});\\
\expone 
& \defsym \,\refone
&& \mid \,\embcache{\symone}{\seq[k]{\refone}}{\expone}
&& \mid \symone(\seq[k]{\expone})
\mid \conone(\seq[k]{\expone});\\
\!\!\!\evalctxone 
& \defsym \hole 
&& \mid \embcache{\symone}{\seq[k]{\refone}}{\evalctxone}
&& \mid \symone(\seq[i-1]{\refone},\evalctxone,\seq[i+1][k]{\expone}) 
\mid \conone(\seq[i-1]{\refone},\evalctxone,\seq[i+1][k]{\expone}).
\end{alignat*}
Here, $\seq[k]{\refone},\refone \in \Loc$, $\symone \in \DS$ and
$\conone \in \CS$ are $k$-ary symbols.  An expression is a term
including references to values that will be stored on the heap.  The
additional construct $\embcache{\symone}{\seq[k]{\refone}}{\expone}$
indicates that the partially evaluated expression $\expone$ descends
from a call $\symone(\seq[k]{\valone})$, with arguments $\valone_i$
stored at location $\refone_i$ on the heap.  A context $\evalctxone$
is an expression with a unique \emph{hole}, denoted as $\hole$, where
all sub-expression to the left of the hole are references pointing to
values. This syntactic restriction is used to implement a
\emph{left-to-right}, \emph{call-by-value} evaluation order.  We
denote by $\evalctxone[\expone]$ the expression obtained by replacing
the hole in $\evalctxone$ by $\expone$.

A \emph{configuration} is a triple $\conf{\hcacheone}{\heapone}{\expone}$ consisting of a \emph{cache} $\hcacheone$, 
\emph{heap} $\heapone$ and expression $\expone$.
Unlike before, the cache $\hcacheone$ consists of pairs of the form $(\symone(\seq[k]{\refone}),\refone)$ 
where instead of values, we store references $\seq[k]{\refone},\refone$ pointing to the heap. 
The heap $\heapone$ is represented as a (multi-rooted) term graph $\heapone$ 
with nodes in $\Loc$ and constructors $\CS$ as labels. 
If $\refone$ is a node of $\heapone$, then we say that $\heapone$ stores at location $\refone$ the 
value $\trepr[\heapone]{\refone}$ obtained by unfolding $\heapone$ starting from location $\refone$.
We keep the heap in a \emph{maximally shared} form, that is, 
$\heapone(\refone_a) = \conone(\seq[k]{\refone}) = \heapone(\refone_b)$ implies $\refone_a = \refone_b$
for two locations $\refone_a,\refone_b$ of $\heapone$.
Thus crucially, values are stored once only, by the following lemma.
\begin{lemma}\label{l:heap:ms}
  Let $\heapone$ be a maximally shared heap with locations $\refone_1,\refone_2$.
  If $\trepr[\heapone]{\refone_1} = \trepr[\heapone]{\refone_2}$ then $\refone_1 = \refone_2$. 
\end{lemma}
The operation $\hmerge{\heapone}{\conone(\seq[k]{\refone})}$, defined as follows,
is used to extend the heap $\heapone$
with a constructor $\conone$ whose arguments point to $\seq[k]{\refone}$, 
retaining maximal sharing. 
Let $\refone_f$ be the first location not occurring in the nodes $\nodes$ of 
$\heapone$ (with respect to an arbitrary, but fixed enumeration on $\Loc$).
For $\seq[k]{\refone} \in \nodes$ we define
\[
\hmerge{\heapone}{\conone(\seq[k]{\refone})} \defsym 
\begin{cases}
  (\heapone,\refone) & \text{if $\heapone(\refone) = \conone(\seq[k]{\refone})$,} \\ 
  (\heapone \cup \{\refone_f \mapsto \conone(\seq[k]{\refone})\},\refone_f) & \text{otherwise.}  
\end{cases}
\]
Observe that the first clause is unambiguous on maximally shared heaps. 

\begin{figure}
  \centering
  \begin{framed}
    \[
    \hfill
    \infer[apply]
    {
      \conf{\hcacheone}{\heapone}{\evalctxone[\symone(\seq[k]{\refone})]}
      \Rrew \conf{\hcacheone}{\heapone}{\evalctxone[\embcache{\symone}{\seq[k]{\refone}}{r\indsubst{\morphone}}]}
    }{
      \deduce{
        \mmatch{\morphone}{\tgraphone}{\symone(\subgraphAt{\heapone}{\refone_1},\dots,\subgraphAt{\heapone}{\refone_k})} 
        \quad\indsubst{\morphone} \defsym \{ \varone \mapsto \morphone(\refone_\varone) \mid \refone_\varone \in \Loc,\ \tgraphone(\refone_\varone) = \varone \in \VS \}
      }
      {\strut(\symone(\seq[k]{\refone}),\refone) \not\in \hcacheone & \symone(\seq[k]{p}) \to r \in \TRSone & \tgraphone \defsym \tree{\symone(\seq[k]{p})}}
    }
    \hfill
    \]
    \[
    \hfill
    \infer[read]
    {
      \conf{\hcacheone}{\heapone}{\evalctxone[\symone(\seq[k]{\refone})]}
      \rread \conf{\hcacheone}{\heapone}{\evalctxone[\refone]}
    }{
      (\symone(\seq[k]{\refone}),\refone) \in \hcacheone 
    }
    \hfill
    \]
    \[
    \hfill
    \infer[store]
    {
      \conf{\hcacheone}{\heapone}{\evalctxone[\embcache{\symone}{\seq[k]{\refone}}{\refone}]}
      \rstore \conf{\hcacheone \cup \{(\symone(\seq[k]{\refone}),\refone)\}}{\heapone}{\evalctxone[\refone]}
    }{}
    \hfill
    \]
    \[
    \hfill
    \infer[merge]
    {
      \conf{\hcacheone}{\heapone}{\evalctxone[\conone(\seq[k]{\refone})]}
      \rmerge \conf{\hcacheone}{\heapone'}{\evalctxone[\refone]}
    }{
      (\heapone',\refone) = \hmerge{\heapone}{\conone(\seq[k]{\refone})}
    }
    \hfill
    \]
  \end{framed}
  \caption{Small Step Semantics with Memoization and Sharing for Program $\prog{\DS}{\CS}{\TRSone}$.}
\label{fig:sss:memo}
\end{figure}

Figure~\ref{fig:sss:memo} collects the small step semantics with respect to a program $\progone = \prog{\DS}{\CS}{\TRSone}$.
We use ${\rsm}$ to abbreviate the relation ${\rread} \cup {\rstore} \cup {\rmerge}$
and likewise we abbreviate ${\Rrew} \cup {\rsm}$ by $\Rrsm$. 
Furthermore, we define ${\rRrsm} \defsym {\rsms} \cdot {\Rrew} \cdot {\rsms}$. 
Hence the \emph{$m$-fold composition} $\rRrsm[m]$ corresponds to a $\Rrsm$-reduction with precisely 
$m$ applications of $\Rrew$. 
\longv{
Throughout the following, we are interested in reductions over \emph{well-formed} configurations:
\begin{definition}\label{d:conf:wf}
  A configuration $\conf{\hcacheone}{\heapone}{\expone}$ is \emph{well-formed} if 
  the following conditions hold.
  \begin{varenumerate}
  \item\label{d:conf:wf:heap} The heap $\heapone$ is maximally shared. 
  \item\label{d:conf:wf:cache} The cache $\hcacheone$ is a function, and \emph{compatible} with 
    $\expone$, that is, if $\embcache{\symone}{\seq[k]{\refone}}{\expone'}$
    occurs as a sub-expression in $\expone$, then $(\symone(\seq[k]{\refone}),\refone) \not\in \hcacheone$ 
    for any $\refone$. 
  \item\label{d:conf:wf:dangling} The configuration contains no \emph{dangling locations}, that is,
    $\heapone(\refone)$ is defined for each location $\refone$ occurring in $\hcacheone$ and $\expone$.
  \end{varenumerate}
\end{definition}
\begin{lemma}\label{l:sss:wf}
  \begin{varenumerate}
    \item\label{l:sss:wf:ctx} If $\conf{\hcacheone}{\heapone}{\evalctxone[\expone]}$ is well-formed 
      then so is $\conf{\hcacheone}{\heapone}{\expone}$.
    \item\label{l:sss:wf:closure} If $\conf{\hcacheone_1}{\heapone_1}{\expone_1} \Rrsm \conf{\hcacheone_2}{\heapone_2}{\expone_2}$  
      and $\conf{\hcacheone_1}{\heapone_1}{\expone_1}$ is well-formed 
      then so is $\conf{\hcacheone_2}{\heapone_2}{\expone_2}$.
  \end{varenumerate}
\end{lemma}
\begin{proof}
  It is not difficult to see that Assertion~\ref{l:sss:wf:ctx} holds. 
  To see that Assertion~\ref{l:sss:wf:closure} holds, 
  fix a well-formed configuration $\conf{\hcacheone_1}{\heapone_1}{\expone_1}$ and
  suppose $\conf{\hcacheone_1}{\heapone_1}{\expone_1} \Rrsm \conf{\hcacheone_2}{\heapone_2}{\expone_2}$. 
  We check that $\conf{\hcacheone_2}{\heapone_2}{\expone_2}$ is well-formed by case 
  analysis on $\Rrsm$.
  \begin{enumerate}
  \item \emph{The heap $\heapone_2$ is maximally shared}: 
    As only the relation $\rmerge$ modifies the heap, it suffices to consider the case
    $\conf{\hcacheone_1}{\heapone_1}{\expone_1} \rmerge \conf{\hcacheone_2}{\heapone_2}{\expone_2}$. 
    Then $(\heapone_2,\refone) = \hmerge{\heapone_1}{\conone(\seq[k]{\refone})}$ for some location $\refone$, 
    and the property follows as $\MERGE$ preserves maximal sharing.
  \item \emph{The cache $\hcacheone_2$ is a function}: 
    It suffices to consider the rules $\rstore$. 
    As immediate consequence of compatibility of $\hcacheone_1$ with $\expone_1$ it follows that $\hcacheone_2$ is a function.
  \item \emph{The cache $\hcacheone_2$ is compatible with $\expone_2$}:
    Only the rules $\Rrew$ and $\rstore$ potentially contradict compatibility. 
    In the former case, the side conditions ensure that $\expone_2$ and $\hcacheone_2$ are compatible, 
    in the latter case compatibility follows trivially from compatibility of 
    $\hcacheone_1$ with $\expone_1$.
  \item \emph{No dangling references}:
    Observe that only rule $\rmerge$ introduces a fresh location. The merge operations 
    guarantees that this location occurs in the heap $\heapone_2$. 
  \end{enumerate}
\end{proof}
From now on, if not mentioned otherwise we will suppose that configurations are well-formed, tacitly employing Lemma~\ref{l:sss:wf}.
}

It is now time to show that the model of computation we have just introduced fits
our needs, namely that it faithfully simulates big-step semantics as in
Figure~\ref{fig:os:memo} (itself a correct implementation of call-by-value evaluation
from Section~\ref{s:basics}). This is proven by first showing how big-step semantics can
be \emph{simulated} by small-step semantics, later proving that the latter is in fact
\emph{deterministic}.

In the following, we denote by $\trepr[\heapone]{\expone}$ the term obtained 
from $\expone$ by following pointers to the heap, ignoring the annotations $\embcache{\symone}{\seq[k]{\refone}}{\cdot}$. 
Formally, we define
\[
 \trepr[\heapone]{\expone} \defsym 
 \begin{cases}
   f(\trepr[\heapone]{\expone_1},\dots,\trepr[\heapone]{\expone_k}) 
   & \text{if $\expone = f(\seq[k]{\expone})$,} \\
   \trepr[\heapone]{\expone'} & \text{if $\expone = \embcache{\symone}{\seq[k]{\refone}}{\expone'}$.}
 \end{cases}
\]
Observe that this definition is well-defined as long as $\heapone$ contains all locations occurring in $\expone$%
\shortv{ (a property that is preserved by $\Rrsm$-reductions)}.
\longv{%
Likewise, we set $\trepr[\heapone]{\hcacheone} \defsym \bigl\{ (\trepr[\heapone]{\expone},\trepr[\heapone]{\refone}) \mid (\expone,\refone) \in \hcacheone \bigr\}$. 

Our simulation result relies on the following two auxiliary lemmas concerning heaps.
The first is based on the observation that in $\Rrsm$-reductions, the heap is monotonically increasing.
\begin{lemma}\label{l:sss:heap}
  If $\conf{\hcacheone_1}{\heapone_1}{\expone_1} \Rrsm \conf{\hcacheone_2}{\heapone_2}{\expone_2}$ then the following properties hold:
  \begin{enumerate}
  \item\label{l:sss:heap:ref} 
    $\trepr[\heapone_2]{\refone} = \trepr[\heapone_1]{\refone}$ for every location $\refone$ of $\heapone_1$; 
  \item\label{l:sss:heap:config} 
    $\trepr[\heapone_2]{\hcacheone_1} = \trepr[\heapone_1]{\hcacheone_1}$ and 
    $\trepr[\heapone_2]{\expone_1} = \trepr[\heapone_1]{\expone_1}$. 
  \end{enumerate}
\end{lemma}
\begin{proof}
  As for any other step the heap remains untouched, the only non-trivial case is 
  $\conf{\hcacheone_1}{\heapone_1}{\evalctxone[\conone(\seq[k]{\refone})]} \rmerge \conf{\hcacheone_1}{\heapone_2}{\evalctxone[\refone]}$
  with $(\heapone_2,\refone) = \hmerge{\heapone_1}{\conone(\seq[k]{\refone})}$. 
  Observe that by definition of $\MERGE$, $\heapone_2(\refone) = \heapone_1(\refone)$ for every $\refone \in \nodes[\heapone_1]$.
  From this Assertion~\ref{l:sss:heap:ref} is easy to establish.
  Assertion~\ref{l:sss:heap:config} follows then by standard inductions 
  on $\hcacheone_1$ and $\evalctxone$, 
  respectively.
\end{proof}

\begin{lemma}\label{l:sss:merge}
  Let $\conf{\hcacheone}{\heapone}{\evalctxone[\valone]}$ be a configuration
  for a value $\valone$. 
  Then $\conf{\hcacheone}{\heapone}{\evalctxone[\valone]} \rmerges \conf{\hcacheone}{\heapone'}{\evalctxone[\refone]}$ 
  with $\trepr[\heapone']{\refone} = \trepr[\heapone]{\valone}$.
\end{lemma}
\begin{proof}
  Note that by assumptions, $\expone$ consist only of constructors and locations. 
  We proof the lemma by induction on the number of constructor symbols in $\expone$. 
  In the base case, $\expone = \refone$ and the lemma trivially holds. 
  For the inductive step, 
  it is not difficult to see that $\expone=\evalctxone'[\conone(\seq[k]{\refone})]$
  for some evaluation context $\evalctxone'$, and hence 
  $\conf{\hcacheone}{\heapone}{\evalctxone[\expone]} \rmerge \conf{\hcacheone}{\heapone'}{\evalctxone[\evalctxone'[\refone]]}$, 
  where $(\heapone',\refone) = \hmerge{\heapone}{\conone(\seq[k]{\refone})}$. 
  Using that $\trepr[\heapone']{\refone} = \trepr[\heapone']{\conone(\seq[k]{\refone})}$
  by definition of $\MERGE$ and Lemma~\eref{l:sss:heap}{config} we conclude 
  $\trepr[\heapone']{\evalctxone[\evalctxone'[\refone]]} = \trepr[\heapone]{\evalctxone[\expone]}$.
  We complete this derivation to the desired form, by induction hypothesis. 
\end{proof}}%
An \emph{initial configuration} is a configuration of the form $\conf{\cacheempty}{\heapone}{\expone}$
with $\heapone$ a maximally shared heap and $\expone = \symone(\seq[k]{\valone})$ an expression 
unfolding to a function call.
Notice that the arguments $\seq[k]{\valone}$ are allowed to contain references to the heap $\heapone$. 
\begin{lemma}[Simulation]\label{l:sss:simulation}
  Let $\conf{\cacheempty}{\heapone}{\expone}$ be an initial configuration.
  If $\ctpair{\cacheempty}{\trepr[\heapone]{\expone}} \reducem[m] \ctpair{\cacheone}{\valone}$ holds for $m \geq 1$ then 
  $\conf{\cacheempty}{\heapone}{\expone} \rRrsm[m] \conf{\hcacheone}{\heaptwo}{\refone}$ 
  for a location $\refone$ in $\heaptwo$ with $\trepr[\heaptwo]{\refone} = \valone$.
\end{lemma}
\longv{
\begin{proof}
  Call a configuration $\conf{\hcacheone}{\heapone}{\expone}$
  \emph{proper} if it is well-formed and $\expone$ does not 
  contain a sub-expression $\embcache{\symone}{\seq[k]{\valone}}{\expone'}$.
  We show the following claim:
  \begin{claim}
    For every proper configuration $\conf{\hcacheone_1}{\heapone_1}{\expone_1}$, 
    $\ctpair{\trepr[\heapone_1]{\hcacheone_1}}{\trepr[\heapone_1]{\expone_1}} \reducem[m] \ctpair{\cacheone_2}{\valone}$
    implies $\conf{\hcacheone_1}{\heapone_1}{\expone_1} \rsms \cdot \rRrsm[m] \conf{\hcacheone_2}{\heapone_2}{\refone}$ 
    with $\ctpair{\trepr[\heapone_2]{\hcacheone_2}}{\trepr[\heapone_2]{\refone}} = \ctpair{\cacheone_2}{\valone}$.
  \end{claim}
  Observe that ${\rsms \cdot \rRrsm[m]} = {\rRrsm[m]}$ whenever $m > 0$. 
  Since an initial configuration is trivially proper, the lemma follows from the claim. 

  To prove the claim, abbreviate the relation ${\rsms \cdot \rRrsm[m]}$ by $\rsl{m}$ for all $m \in \N$. 
  Below, we tacitly employ ${\rsl{m_1} \cdot \rsl{m_2}} = {\rsl{m_1 + m_2}}$ for all $m_1,m_2\in\N$. 
  The proof is by induction on the deduction 
  $\Pi$ of the statement
  $\ctpair{\trepr[\heapone]{\hcacheone_1}}{\trepr[\heapone_1]{\expone}} \reducem[m] \ctpair{\cacheone}{\valone}$. 
  \begin{enumerate}
  \item Suppose that the last rule in $\Pi$ has the form
    \[
    \infer
    {\ctpair{\cacheone_0}{\conone(\seq[k]{\termone})} \reducem[m] \ctpair{\cacheone_k}{\conone(\seq[k]{\valone})}}
    {\conone \in \CS & \ctpair{\cacheone_{i-1}}{\termone_i} \reducem[m_i] \ctpair{\cacheone_i}{\valone_i} & m = \sum_{i=1}^k m_i}
    \]
    Fix a proper configuration $\conf{\hcacheone_0}{\heapone_0}{\expone_0}$ unfolding to 
    $\ctpair{\cacheone_0}{\conone(\seq[k]{\termone})}$.
    Under these assumptions, either
    $\expone_0$ is a location or $\expone_0 = \conone(\seq[k]{\expone})$.
    The former case is trivial, as then $\termone$ is a value and thus $m=0$. 
    Hence suppose $\expone_0 = \conone(\seq[k]{\expone})$. 
    We first show that for all $i \leq k$, 
    \begin{equation}
      \label{eq:sss:simulation:1}
      \tag{\dag}
      \conf{\hcacheone_0}{\heapone_0}{\conone(\seq[k]{\expone})} 
      \rsl{\sum_{j=1}^i m_j} 
      \conf{\hcacheone_i}{\heapone_i}{\conone(\seq[i]{\refone},\seq[i+1][k]{\expone})}
      \tkom
    \end{equation}
    for a configuration $\conf{\hcacheone_i}{\heapone_i}{\conone(\seq[i]{\refone},\seq[i+1][k]{\expone})}$
    which unfolds to the configuration $\ctpair{\cacheone_i}{\conone(\seq[i]{\valone},\seq[i+1][k]{\termone})}$. 
    The proof is by induction on $i$, we consider the step from $i$ to $i+1$. 
    Induction hypothesis yields a well-formed configuration 
    $\conf{\hcacheone_i}{\heapone_i}{\evalctxone[\expone_{i+1}]}$
    for $\evalctxone = \conone(\seq[i]{\refone},\hole,\seq[i+2][k]{\expone})$
    reachable by a Derivation~\eqref{eq:sss:simulation:1}. 
    As the configuration $\conf{\hcacheone_i}{\heapone_i}{\expone_{i+1}}$ unfolds to $\ctpair{\cacheone_i}{\termone_{i+1}}$, 
    the induction hypothesis of the claim on the assumption $\ctpair{\cacheone_i}{\termone_{i+1}} \reducem[m_{i+1}] \ctpair{\cacheone_{i+1}}{\valone_{i+1}}$ yields
    $\conf{\hcacheone_i}{\heapone_i}{\expone_{i+1}} \rsl{m_{i+1}} \conf{\hcacheone_{i+i}}{\heapone_{i+1}}{\refone_{i+1}}$
    where the resulting configuration unfolds to $\ctpair{\cacheone_{i+1}}{\valone_{i+1}}$. 
    As a consequence, its not difficult to see that also 
    $\conf{\hcacheone_i}{\heapone_i}{\evalctxone[\expone_{i+1}]} 
    \rsl{m_{i+1}} 
    \conf{\hcacheone_{i+1}}{\heapone_{i+1}}{\evalctxone[\refone_{i+1}]}$ holds. 
    Since $\trepr[\heapone_{i+1}]{\refone_{i+1}} = \valone_{i+1}$
    and $\trepr[\heapone_{i}]{\evalctxone[\expone_{i+1}]} = \conone(\seq[i]{\valone},\termone_{i+1},\seq[i+2][k]{\termone})$,
    using Lemma~\eref{l:sss:heap}{config} on the last equality 
    it is not difficult to see that 
    $\trepr[\heapone_{i+1}]{\evalctxone[\refone_{i+1}]} = \conone(\seq[i]{\valone},\valone_{i+1},\seq[i+2][k]{\termone})$.
    As we already observed $\trepr[\heapone_{i+1}]{\hcacheone_{i+1}} = \cacheone_{i+1}$,
    we conclude the Reduction~\eqref{eq:sss:simulation:1}. 

    In total, we thus obtain a reduction
    $\conf{\hcacheone_0}{\heapone_0}{\conone(\seq[k]{\expone})} \rsl{m} \conf{\hcacheone_k}{\heapone_k}{\conone(\seq[k]{\refone})}$
    where $m = \sum_{i=1}^k m_i$
    and $\conf{\hcacheone_k}{\heapone_k}{\conone(\seq[k]{\refone})}$ is a well-formed, 
    in fact proper, configuration which unfolds to $\ctpair{\cacheone_k}{\conone(\seq[k]{\valone})}$.
    Employing $\rsl{m} \cdot {\rmerges} = {\rsl{m}}$ 
    we conclude the case with Lemma~\ref{l:sss:merge}.

  \item Suppose that the last rule in $\Pi$ has the form
    \[
    \infer
    {\ctpair{\cacheone_0}{\symone(\seq[k]{\termone})} \reducem[m] \ctpair{\cacheone_{k+1}}{\valone}}
    { \ctpair{\cacheone_{i-1}}{\termone_i} \reducem[m_i] \ctpair{\cacheone_i}{\valone_i} 
      & \ctpair{\cacheone_k}{\symone(\seq[k]{\valone})} \reducem[n] \ctpair{\cacheone_{k+1}}{\valone} 
      & m = n + \sum_{i=1}^k m_i
    }
    \]
    Fix a proper configuration $\conf{\hcacheone_0}{\heapone_0}{\expone_0}$ unfolding to 
    $\ctpair{\cacheone_0}{\symone(\seq[k]{\termone})}$.
    By induction on $k$, exactly as in the previous case, 
    we obtain a proper configuration $\conf{\hcacheone_k}{\heapone_k}{\symone(\seq[k]{\refone})}$
    unfolding to $\ctpair{\cacheone_k}{\symone(\seq[k]{\valone})}$
    with 
    \[
    \conf{\hcacheone_0}{\heapone_0}{\expone_0} 
    \rsl{\sum_{i=1}^k m_i}
    \conf{\hcacheone_k}{\heapone_k}{\symone(\seq[k]{\refone})}
    \tpkt
    \]
    The induction hypothesis also yields configuration
    $\conf{\hcacheone_{k+1}}{\heapone_{k+1}}{\refone}$ unfolding to 
    $\ctpair{\cacheone_{k+1}}{\valone}$ with
    $\conf{\hcacheone_k}{\heapone_k}{\symone(\seq[k]{\refone})} \rsl{n} \conf{\hcacheone_{k+1}}{\heapone_{k+1}}{\refone}$. 
    Summing up we conclude the case.
  \item Suppose that the last rule in $\Pi$ has the form
  \[
  \infer
    {\ctpair{\cacheone}{\symone(\seq[k]{\valone})} \reducem[0] \ctpair{\cacheone}{\valone}}
    {(\symone(\seq[k]{\valone}),\valtwo) \in \cacheone}
  \]
  Consider a proper configuration $\conf{\hcacheone}{\heapone}{\expone}$ 
  that unfolds to $\ctpair{\cacheone}{\symone(\seq[k]{\valone})}$. 
  Then $\expone = \symone(\seq[k]{\expone})$, 
  and using  $k$ applications of Lemma~\ref{l:sss:merge}
  we construct a reduction
  \[
  \conf{\hcacheone}{\heapone}{\symone(\seq[k]{\expone})} 
  \rmerges \conf{\hcacheone}{\heapone_1}{\symone(\refone_1,\seq[2][k]{\expone})}
  \rmerges \cdots
  \rmerges \conf{\hcacheone}{\heapone_k}{\symone(\seq[k]{\refone})}
  \tkom
  \]
  with $\conf{\hcacheone}{\heapone_k}{\symone(\seq[k]{\refone})}$
  unfolding to $\ctpair{\cacheone}{\symone(\seq[k]{\valone})}$.
  Lemma~\ref{l:heap:ms} and the assumption on $\cacheone = \trepr[\heapone_k]{\hcacheone}$ implies that there 
  exists a \emph{unique} pair $(\symone(\seq[k]{\refone}),\refone) \in \hcacheone$ 
  with $\trepr[\heapone_k]{\symone(\seq[k]{\refone})} = \symone(\seq[k]{\valone})$ and $\trepr[\heapone_k]{\refone} = \valone$. 
  Thus overall
  \[
  \conf{\hcacheone}{\heapone}{\expone} 
  = \conf{\hcacheone}{\heapone}{\symone(\seq[k]{\expone})}
  \rmerges \conf{\hcacheone}{\heapone_k}{\symone(\seq[k]{\refone})} 
  \rread \conf{\hcacheone}{\heapone_k}{\refone}
  \tkom
  \]
  where $\conf{\hcacheone}{\heapone_k}{\refone}$ unfolds to $\ctpair{\cacheone}{\valone}$.
  Using ${\rmerges \cdot \rread} \subseteq \rsl{0}$ we conclude the case. 
  \item 
    Finally, suppose that the last rule in $\Pi$ has the form
    \[
    \infer
    {\ctpair{\cacheone}{\symone(\seq[k]{\valone})} \reducem[m+1] \ctpair{\cachetwo \cup \{(\symone(\seq[k]{\valone}),\valone)\}}{\valone}}
    { (\symone(\seq[k]{\valone}),\valone) \not\in \cacheone 
      & \symone(\seq[k]{l}) \to r \in \TRSone 
      & \forall i.\ l_i\sigma = \valone_i
      & \ctpair{\cacheone}{r\sigma} \reducem[m] \ctpair{\cachetwo}{\valone}
    }
    \]
    Fix a proper configuration $\conf{\hcacheone}{\heapone}{\expone}$ 
    that unfolds to $\ctpair{\cacheone}{\symone(\seq[k]{\valone})}$, 
    in particular $\expone = \symone(\seq[k]{\expone})$. 
    As above, we see 
    $\conf{\hcacheone}{\heapone}{\expone} 
    \rmerges \conf{\hcacheone}{\heapone_k}{\symone(\seq[k]{\refone})}$ 
    for a configuration $\conf{\hcacheone}{\heapone_k}{\symone(\seq[k]{\refone})}$
    also unfolding to $\ctpair{\cacheone}{\symone(\seq[k]{\valone})}$. 
    As $(\symone(\seq[k]{\valone}),\valone) \not\in \cacheone$, we have 
    $(\symone(\seq[k]{\refone}),\refone) \not\in \hcacheone_k$ for any location $\refone$, 
    by Lemma~\ref{l:heap:ms}. Since Proposition~\ref{p:tg:match} on the assumption 
    yields
    $\mmatch{\morphone}{\tree{\symone(\seq[k]{l})}}{\symone(\subgraphAt{\heapone}{\refone_1},\dots,\subgraphAt{\heapone}{\refone_k})}$
    for a matching morphism $\morphone$, in total we obtain
    \[
    \conf{\hcacheone}{\heapone}{\expone} 
    \rmerges 
    \conf{\hcacheone}{\heapone}{\symone(\seq[k]{\refone})}
    \Rrew \conf{\hcacheone}{\heapone_k}{\embcache{\symone}{\seq[k]{\refone}}{r\indsubst{\morphone}}}
    \tpkt
    \]
    Note that by Proposition~\ref{p:tg:match} the substitution $\sigma$ and 
    induced substitution $\indsubst{\morphone}$ satisfy
    $\sigma(\varone) = \trepr[\heapone_k]{\indsubst{\morphone}(\varone)}$ for all variables $\varone$ in $\termone$.
    Hence by a standard induction on $r$, 
    $\trepr[\heapone_k]{r\indsubst{\morphone}} = r\sigma$ follows. 
    We conclude that $\conf{\hcacheone}{\heapone_k}{\embcache{\symone}{\seq[k]{\refone}}{r\indsubst{\morphone}}}$ 
    unfolds to $\ctpair{\cacheone}{r\sigma}$. 
    Thus the induction hypothesis yields a well-formed configuration 
    $\conf{\hcacheone'}{\heaptwo}{\refone}$ 
    unfolding to $\ctpair{\cacheone'}{\valone}$ with
    $\conf{\hcacheone}{\heapone_k}{r\indsubst{\morphone}}
    \rsl{m} \conf{\hcacheone'}{\heaptwo}{\refone}$.
    Thus
    \begin{alignat*}{3}
      \conf{\hcacheone}{\heapone_k}{\embcache{\symone}{\seq[k]{\refone}}{r\indsubst{\morphone}}}
      & \rsl{m} &&\conf{\hcacheone'}{\heaptwo}{\embcache{\symone}{\seq[k]{\refone}}{\refone}} \\
      & \rstore &&\conf{\hcacheone' \cup \{ (\symone(\seq[k]{\refone}),\refone) \}}{\heaptwo}{\refone}
      \tpkt
    \end{alignat*}
    Using that $\trepr[\heaptwo]{\refone} = \valone$
    and 
    $\trepr[\heapone_k]{\symone(\seq[k]{\refone})} = \symone(\seq[k]{\valone})$, 
    Lemma~\ref{l:sss:heap} yields
    \[
    \trepr[\heaptwo]{\hcacheone' \cup \{ (\symone(\seq[k]{\refone}),\refone) \}} = \cacheone' \cup \{ (\symone(\seq[k]{\valone}),\valone) \}
    \tpkt
    \]
    Putting things together, 
    employing ${\rmerges} \cdot {\Rrew} \subseteq {\rsl{1}}$ and ${\rsl{m}} \cdot {\rstore} = {\rsl{m}}$ 
    we conclude $\conf{\hcacheone}{\heapone}{\expone} \rsl{m+1} \conf{\hcacheone' \cup \{ (\symone(\seq[k]{\refone}),\refone) \}}{\heaptwo}{\refone}$,
    where the resulting configuration unfolds to $\ctpair{\cacheone' \cup \{ (\symone(\seq[k]{\valone}),\valone) \}}{\valone}$. 
    We conclude this final case. 
  \end{enumerate}
\end{proof}}

The next lemma shows that the established simulation is \emph{unique}, that is, there is exactly one derivation 
$\conf{\cacheempty}{\heapone}{\expone} \rRrsm[m] \conf{\hcacheone}{\heaptwo}{\refone}$. 
Here, a relation $\rew$ is called \emph{deterministic on a set $A$} if $b_1 \wer a \rew b_2$ implies $b_1 = b_2$ for all $a \in A$. 
\begin{lemma}[Determinism]\label{l:sss:deterministic}
\shortlongv{%
  The relation $\Rrsm$ is deterministic on all configurations reachable from initial configurations. 
}
{%
  \begin{enumerateenv}
    \item\label{l:sss:deterministic:seperate} The relations $\Rrew$, $\rread$, $\rstore$ and $\rmerge$ are deterministic on well-formed configurations. 
    \item\label{l:sss:deterministic:combined} The relation $\Rrsm$ is deterministic on well-formed configurations. 
  \end{enumerateenv}
}
\end{lemma}
\longv{
\begin{proof}
  For Assertion~\ref{l:sss:deterministic:seperate}, fix ${\rew[r]} \in \{{\Rrew},{\rread},{\rstore},{\rmerge}\}$. 
  Let $\conf{\hcacheone}{\heapone}{\expone}$ be a well-formed configuration. 
  We show that any \emph{peak} 
  $\conf{\hcacheone_1}{\heapone_1}{\expone_1} \wer[r] \conf{\hcacheone}{\heapone}{\expone} \rew[r] \conf{\hcacheone_2}{\heapone_2}{\expone_2}$
  is \emph{trivial}, i.e. $\conf{\hcacheone_1}{\heapone_1}{\expone_1} = \conf{\hcacheone_2}{\heapone_2}{\expone_2}$.
  Observe that independent on $\rew[r]$, the evaluation context $\evalctxone$ 
  in the corresponding rule is unique.
  From this, we conclude the assertion by case analysis on $\rew[r]$. 
  The non-trivial cases are ${\rew[r]} = {\Rrew}$ and ${\rew[r]} = {\rread}$. 
  In the former case, we conclude using that rules in $\TRSone$ are non-overlapping,
  tacitly employing Proposition~\ref{p:tg:match}. 
  The latter case we conclude using that $\hcacheone$ is well-formed. 
  
  Finally, for Assertion~\ref{l:sss:deterministic:combined} consider a peak 
  $\conf{\hcacheone_1}{\heapone_1}{\expone_1} \wer[r_1] \conf{\hcacheone}{\heapone}{\expone} \rew[r_2] \conf{\hcacheone_2}{\heapone_2}{\expone_2}$
  ${\rew[r_1]},{\rew[r_2]} \in \{{\Rrew},{\rread},{\rstore},{\rmerge}\}$. 
  We show that this peak is trivial by induction on the expression $\expone$. 
  By the previous assertion, it suffices to consider only the case ${\rew[r_1]} \not= {\rew[r_2]}$. 

  The base case constitutes of the cases 
  (i)~$\expone = \symone(\seq[k]{\refone})$, 
  (ii)~$\expone = \embcache{\symone}{\seq[k]{\refone}}{\refone}$ and
  (iii)~$\expone = \conone(\seq[k]{\refone})$. 
  The only potential peak can occurs in case (i) between relations $\Rrew$ and $\rread$.
  Here, a non-trivial peak is prohibited by the pre-conditions put on $\heapone$. 
  For the inductive step, we consider a peak 
  \[
  \conf{\hcacheone_1}{\heapone_1}{\evalctxone[\expone'_1]} 
  \wer[r_1] \conf{\hcacheone}{\heapone}{\evalctxone[\expone']} 
  \rew[r_2] \conf{\hcacheone_2}{\heapone_2}{\evalctxone[\expone'_2]}
  \tkom
  \]
  where $\expone = \evalctxone[\expone']$ for a 
  context $\evalctxone$.
  As we thus have a peak
  $\conf{\hcacheone_1}{\heapone_1}{\expone'_1} \wer[r_1] \conf{\hcacheone}{\heapone}{\expone'} \rew[r_2] \conf{\hcacheone_2}{\heapone_2}{\expone'_2}$, 
  which by induction hypothesis is trivial, we conclude the assertion.
\end{proof}}

\begin{theorem}\label{t:sss:simulation}
  Suppose $\ctpair{\cacheempty}{\symone(\seq[k]{\valone})} \reducem[m]
  \ctpair{\cacheone}{\valone}$ holds for a reducible term
  $\symone(\seq[k]{\valone})$.  Then for each initial configuration
  $\conf{\cacheempty}{\heapone}{\expone}$ with $\trepr[\heapone]{\expone} = \symone(\seq[k]{\valone})$,
  there exists a unique sequence
  $\conf{\cacheempty}{\heapone}{\expone} \rRrsm[m] \conf{\hcacheone}{\heaptwo}{\refone}$ 
  for a location $\refone$ in $\heaptwo$ with $\trepr[\heaptwo]{\refone} = \valone$.
\end{theorem}
\begin{proof}
  As $\symone(\seq[k]{\valone})$ is reducible, it follows that $m
  \geq 1$.  Hence the theorem follows from
  Lemma~\ref{l:sss:simulation} and Lemma~\ref{l:sss:deterministic}.
\end{proof}

\paragraph*{Invariance}
Theorem~\ref{t:sss:simulation} tells us that a term-based semantics (in which
sharing is \emph{not} exploited) can be simulated step-by-step by another,
more sophisticated, graph-based semantics. The latter's advantage is that
each computation step does not require copying, and thus does not increase
the size of the underlying configuration too much. This is the key 
observation towards \emph{invariance}: the number of reduction
step is a sensible cost model from a complexity-theoretic perspective. Precisely this
will be proven in the remaining of the section.

Define the \emph{size} $\size{\expone}$ of an expression recursively by $\size{\refone} \defsym 1$, 
$\size{f(\seq[k]{\expone})} \defsym 1 + \sum_{i=1}^k \size{\expone_i}$ and 
$\size{\embcache{\symone}{\seq[k]{\refone}}{\expone}} \defsym 1 + \size{\expone}$. 
In correspondence we define the \emph{weight} $\wexp{\expone}$ by ignoring locations, 
i.e.\ $\wexp{\refone} \defsym 0$.
Recall that a reduction $\conf{\hcacheone_1}{\heapone_1}{\expone_1} \rRrsm[m] \conf{\hcacheone_2}{\heapone_2}{\expone_2}$
consists of $m$ applications of $\Rrew$, all interleaved by $\rsm$-reductions. 
As a first step, we thus estimate the overall length of the reduction $\conf{\hcacheone_1}{\heapone_1}{\expone_1} \rRrsm[m] \conf{\hcacheone_2}{\heapone_2}{\expone_2}$ 
in $m$ and the size of $\expone_1$. 
Set $\Delta \defsym \max\{ \size{r} \mid l \to r \in \TRSone \}$. 
The following serves as an intermediate lemma.
\begin{lemma}\label{l:weight}
  The following properties hold:
  \begin{enumerate}
  \item\label{l:weight:rsm} 
    If $\conf{\hcacheone_1}{\heapone_1}{\expone_1} \rsm \conf{\hcacheone_2}{\heapone_2}{\expone_2}$ 
    then $\wexp{\expone_2} < \wexp{\expone_1}$.
  \item\label{l:weight:R} 
    If $\conf{\hcacheone_1}{\heapone_1}{\expone_1} \Rrew \conf{\hcacheone_2}{\heapone_2}{\expone_2}$ 
    then $\wexp{\expone_2} \leq \wexp{\expone_1} + \Delta$.
  \end{enumerate}
\end{lemma}
\longv{
\begin{proof}
  The first assertion follows by case analysis on $\rsm$. 
  For the second, suppose 
  $\conf{\hcacheone_1}{\heapone_1}{\expone_1} \rsm \conf{\hcacheone_2}{\heapone_2}{\expone_2}$
  where $\expone_1 = \evalctxone[\symone(\seq[k]{\refone})]$ and 
  $\expone_2 = \evalctxone[\embcache{\symone}{\seq[k]{\refone}}{r\indsubst{\morphone}}]$
  for a rule $\symone(\seq[k]{l}) \to r \in \TRSone$. 
  Observe that since the substitution $\indsubst{\morphone}$ replaces variables by locations, 
  $\Delta \geq \size{r} = \size{r\indsubst{\morphone}} \geq \wexp{r\indsubst{\morphone}}$ holds.
  Consequently,  
  \begin{align*}
  \wexp{\symone(\seq[k]{\refone})} + \Delta 
  & \geq 1 + \wexp{r\indsubst{\morphone}} 
  = \wexp{\embcache{\symone}{\seq[k]{\refone}}{r\indsubst{\morphone}}} \tpkt
  \end{align*}
  Then the assertion follows by a standard induction on $\evalctxone$. 
\end{proof}}

Then essentially an application of the \emph{weight gap principle}~\cite{HM:IJCAR:08}, 
a form of \emph{amortized} cost analysis, binds the overall length of an $\rRrsm[m]$-reduction suitably.
\begin{lemma}\label{l:sss:relapprox}
  If $\conf{\hcacheone_1}{\heapone_1}{\expone_1} \rRrsm[m] \conf{\hcacheone_2}{\heapone_2}{\expone_2}$
  then $\conf{\hcacheone_1}{\heapone_1}{\expone_1} \Rrsm[n] \conf{\hcacheone_2}{\heapone_2}{\expone_2}$
  for $n \leq (1 + \Delta) \cdot m + \wexp{\expone}$ 
  and $\Delta \in \N$ a constant depending only on $\progone$. 
\end{lemma}
\longv{
\begin{proof}
  For a configuration $\confone = \conf{\hcacheone}{\heapone}{\expone}$
  define $\wexp{\confone} \defsym \wexp{\expone}$ 
  and let $\Delta$ be defined as in Lemma~\ref{l:weight}. 
  Consider $\conf{\hcacheone_1}{\heapone_1}{\expone_1} \rRrsm[m] \conf{\hcacheone_2}{\heapone_2}{\expone_2}$
  which can be written as a reduction
  \begin{equation}
    \label{eq:invariance}
    \tag{\ddag}
    \conf{\hcacheone_1}{\heapone_1}{\expone_1} = \confone_0
    \rsm[n_0] \conftwo_0
    \Rrew \confone_1
    \rsm[n_1] \conftwo_1
    \Rrew \cdots 
    \rsm[n_m] \conftwo_m
    \tkom
  \end{equation}
  of length $n \defsym m + \sum_{i=0}^k n_i$. 
  Lemma~\ref{l:weight} yields
  (i)~$n_i \leq \wexp{\confone_i} - \wexp{\conftwo_i}$ for all $0 \leq i \leq m$; and 
  (ii)~$\wexp{\confone_{i+1}} - \wexp{\conftwo_i} \leq \Delta$ for all $0 \leq i < m$.
  Hence overall, the Reduction~\eqref{eq:invariance} is of length
  \begin{align*}
      n
      & \leq m + (\wexp{\confone_0} - \wexp{\conftwo_0}) + \cdots + (\wexp{\confone_m} - \wexp{\conftwo_m})\\
      & = m + \wexp{\confone_0} + (\wexp{\confone_1} - \wexp{\conftwo_0}) + \cdots + (\wexp{\confone_m} - \wexp{\conftwo_{m-1}}) - \wexp{\conftwo_m} \\
      & \leq m + \wexp{\confone_0} + m \cdot \Delta\\
      & = (1 + \Delta) \cdot m + \wexp{\expone} \tpkt
  \end{align*}
  The lemma follows.
\end{proof}}

Define the size of a configuration $\size{\conf{\hcacheone}{\heapone}{\expone}}$ as the sum of the sizes of its components. 
Here, the size $\size{\hcacheone}$ of a cache $\hcacheone$ is defined as its cardinality, 
similar, the size $\size{\heapone}$ of a heap is defined as the cardinality of its set of nodes. 
Notice that a configuration $\conf{\hcacheone}{\heapone}{\expone}$
can be straight forward encoded within logarithmic space-overhead as a string $\enc{\conf{\hcacheone}{\heapone}{\expone}}$, 
i.e.\ the length of the string $\enc{\conf{\hcacheone}{\heapone}{\expone}}$ is bounded by a function in 
$\bigO(\log(n) \cdot n)$ in $\size{\conf{\hcacheone}{\heapone}{\expone}}$, 
using constants to encode symbols and an encoding of locations logarithmic in $\size{\heapone}$. 
Crucially, a step in the small-step semantics increases the size of a configuration only by a constant. 

\begin{lemma}\label{l:size}
  If $\conf{\hcacheone_1}{\heapone_1}{\expone_1} \Rrsm \conf{\hcacheone_2}{\heapone_2}{\expone_2}$ 
  then $\size{\conf{\hcacheone_2}{\heapone_2}{\expone_2}} \leq \size{\conf{\hcacheone_1}{\heapone_1}{\expone_1}} + \Delta$.
\end{lemma}
\longv{
\begin{proof}
  The lemma follows by case analysis on the rule applied in
  $\conf{\hcacheone_1}{\heapone_1}{\expone_1} \Rrsm \conf{\hcacheone_2}{\heapone_2}{\expone_2}$, 
  using $1 \leq \Delta$. 
\end{proof}}

\begin{theorem}\label{t:sss:invariance}
  There exists a polynomial $\ofdom{p}{\N \times \N \rightarrow \N}$ such that for 
  every initial configuration $\conf{\cacheempty}{\heapone_1}{\expone_1}$, 
  a configuration $\conf{\hcacheone_2}{\heapone_2}{\expone_2}$ with 
  $\conf{\cacheempty}{\heapone_1}{\expone_1} \rRrsm[m] \conf{\hcacheone_2}{\heapone_2}{\expone_2}$
  is computable from $\conf{\cacheempty}{\heapone_1}{\expone_1}$ in time $\mathop{p}(\size{\heapone_1} + \size{\expone_1},m)$. 
\end{theorem}
\begin{proof}
  It is tedious, but not difficult to show that 
  the function which implements a step $\confone \Rrsm \conftwo$, 
  i.e.\ which maps $\enc{\confone}$ to $\enc{\conftwo}$,
  is computable in polynomial time in $\enc{\confone}$, and thus in the size $\size{\confone}$ of the configuration $\confone$. 
  Iterating this function at most 
  $n \defsym (1 + \Delta) \cdot m + \size{\conf{\cacheempty}{\heapone_1}{\expone_1}}$ 
  times on input $\enc{\conf{\cacheempty}{\heapone_1}{\expone_1}}$, yields 
  the desired result $\enc{\conf{\hcacheone_2}{\heapone_2}{\expone_2}}$ by Lemma~\ref{l:sss:relapprox}.
  Since each iteration increases the size of a configuration by at most the constant $\Delta$ (Lemma~\ref{l:size}), 
  in particular the size of each intermediate configuration is bounded 
  by a linear function in $\size{\conf{\cacheempty}{\heapone_1}{\expone_1}} = \size{\heapone_1} + \size{\expone_1}$ and $n$,
  the theorem follows. 
\end{proof}
Combining Theorem~\ref{t:sss:simulation} and Theorem~\ref{t:sss:invariance} we thus obtain the following.
\begin{theorem}\label{t:invariance}
  There exists a polynomial $\ofdom{p}{\N \times \N \rightarrow \N}$ such that for 
  $\ctpair{\cacheempty}{\symone(\seq[k]{\valone})} \reducem[m] \ctpair{\cacheone}{\valone}$,
  the value $\valone$ represented as DAG is computable from 
  $\seq[k]{\valone}$ in time $\mathop{p}(\sum_{i=1}^k\size{\valone_i},m)$.
\end{theorem}

Theorem~\ref{t:invariance} thus confirms that the cost $m$ of a
reduction $\ctpair{\cacheempty}{\symone(\seq[k]{\valone})} \reducem[m]
\ctpair{\cacheone}{\valone}$ is a suitable cost measure.  In other
words, the \emph{memoized runtime complexity} of a function $\symone$,
relating input size $n \in \N$ to the maximal cost $m$ of evaluation
$\symone$ on arguments $\seq[k]{\valone}$ of size up to $n$,
i.e.\ $\ctpair{\cacheempty}{\symone(\seq[k]{\valone})} \reducem[m]
\ctpair{\cacheone}{\valone}$ with $\sum_{i=1}^k\size{\valone_i} \leq
n$, is an \emph{invariant cost model}.

\begin{example}[Continued from Example~\ref{ex:program}]\label{ex:invariance}
  Reconsider the program $\prograbbits$ and the evaluation of a call
  $\srabbits(\csuc^n(\czero))$ which results in the genealogical tree
  $\valone_n$ of height $n \in \N$ associated with \emph{Fibonacci's
    rabbit problem}.  Then one can show that
  $\srabbits(\csuc^n(\czero)) \reducem[m] \valone_n$ with $m \leq
  2\cdot n + 1$.  Crucially here, the two intermediate functions
  $\sadults$ and $\sbabies$ defined by simultaneous recursion are
  called only on proper subterms of the input $\csuc^n(\czero)$, hence
  in particular the rules defining $\sadults$ and $\sbabies$
  respectively, are unfolded at most $n$ times.  As a consequence of
  the bound on $m$ and Theorem~\ref{t:invariance} we obtain that the
  function $\frabbits$ from the introduction is polytime computable.
\end{example}

\begin{remark}  
  Strictly speaking, our DAG representation of a value $\valone$, viz
  the part of the final heap reachable from a corresponding location
  $\refone$, is not an encoding in the classical, complexity theoretic
  setting.  Different computations resulting in the same value
  $\valone$ can produce different DAG representations of $\valone$,
  however, these representations differ only in the naming of
  locations.  Even though our encoding can be exponentially compact in
  comparison to a linear representation without sharing, it is not
  exponentially more \emph{succinct} than a reasonable encoding for
  graphs (e.g.\ representations as circuits, see \citet{Papa}).  In
  such succinct encodings not even equality can be decided in
  polynomial time.  Our form of representation does clearly not fall
  into this category.  In particular, in our setting it can be easily
  checked in polynomial time that two DAGs represent the same value.
 \end{remark}

\section{GRSR is Sound for Polynomial Time}
Sometimes (e.g., in~\cite{BC:CC:92}), the first step towards a proof
of soundness for ramified recursive systems consists in giving a
proper bound precisely relating the size of the result and the size of
the inputs. More specifically, if the result has tier $n$, then the
size of it depends polynomially on the size of the inputs of tier
higher than $n$, but only \emph{linearly}, and in very restricted way,
on the size of inputs of tier $n$.  Here, a similar result holds, but
size is replaced by \emph{minimal shared size}. 

The \emph{minimal shared size} $\ssz{\seq[k]{\valone}}$ for a \emph{sequence} of 
elements $\seq[k]{\valone} \in \faone$ is defined as the number of
subterms in $\seq[k]{\valone}$, i.e.\ the cardinality of the set
$\bigcup_{1 \leq i \leq k} \subterms{\valone_i}$. Then
$\ssz{\seq[k]{\valone}}$ corresponds to the number of locations
necessary to store the values $\seq[k]{\valone}$ on a heap (compare
Lemma~\ref{l:heap:ms}). 
If $\multyone$ is the
expression $\faone_{n_1}\times\ldots\times\faone_{n_m}$, $n$ is a
natural number, and $\vec{\termone}$ is a sequence of $m$ terms, then
$\ssz{\vec{\termone}}_\multyone^{>n}$ is defined to be
$\ssz{\termone_{i_1},\ldots,\termone_{i_k}}$ where $i_1,\ldots,i_k$
are precisely those indices such that $n_{i_1},\ldots,n_{i_k}>n$. Similarly
for $\ssz{\vec{\termone}}_\multyone^{=n}$.
\begin{proposition}[Max-Poly]
  If $\tj{\funone}{\multyone}{\faone_n}$, then there is a polynomial
  $\ofdom{p_\funone}{\N\to\N}$ such that
  $\ssz{\funone(\vec{\valone})}\leq
  \ssz{\vec{\valone}}_\multyone^{=n}+p_\funone(\ssz{\vec{\valone}}_\multyone^{>n})$.
\end{proposition}
Once we know that ramified recursive definitions are not too
fast-growing for the minimal shared size, we know that all terms
around do not have a too-big minimal shared size. As a consequence:
\begin{proposition}\label{prop:bound}
  If $\tj{\funone}{\multyone}{\faone_n}$, then there is a polynomial
  $\ofdom{p_\funone}{\N\to\N}$ such that for every $\valone$,
  $\ctpair{\cacheempty}{\symone(\vec{\valone})} \reducem[m] \ctpair{\cacheone}{\valone}$,
  with $m\leq p_\funone(\ssz{\vec{\valone}})$.
\end{proposition}
The following, then, is just a corollary of Proposition~\ref{prop:bound} and Invariance (Theorem~\ref{t:invariance}). 
\begin{theorem}\label{t:sound}
  Let $\ofdom{\funone}{\faone_{p_1}\times\ldots\times\faone_{p_k} \rightarrow \faone_{m}}$ be a function defined by 
  general ramified simultaneous recursion. There exists then a polynomial
  $\ofdom{p_{\funone}}{\N^k \to \N}$ such that for all inputs
  $\seq[k]{\valone}$, a DAG representation of
  $\funone(\seq[k]{\valone})$ is computable in time
  $p_{\funone}(\size{\valone_1},\ldots,\size{\valone_n})$.
\end{theorem}

\begin{example}[Continued from Example~\ref{ex:invariance}]
  In Example~\ref{ex:algebra} we indicated that the function $\ofdom{\srabbits}{\fanat \rightarrow \farabbits}$ 
  from Section~\ref{s:examples} is definable by GRSR\@. 
  As a consequence of Theorem~\ref{t:sound}, it is computable in polynomial time, 
  e.g.\ on a Turing machine. Similar, we can prove the function $\fun{tree}$ 
  from Section~\ref{s:examples} polytime computable. 
\end{example}


\section{Conclusion}
In this work we have shown that simultaneous ramified recurrence
on generic algebras is sound for polynomial time, resolving a long-lasting open problem in
implicit computational complexity theory. 
We believe that with this work we have reached the \emph{end of a quest}. 
Slight extensions, e.g.\ the inclusion of \emph{parameter substitution}, 
lead outside polynomial time as soon as simultaneous recursion over trees is permissible. 

Towards our main result, we introduced the notion of memoized runtime complexity,
and we have shown that this cost model is invariant under polynomial time.
Crucially, we use a compact DAG representation of values to control duplication, 
and tabulation to avoid expensive re-computations. 
To the authors best knowledge, our work is the first where sharing and memoization is reconciled, 
in the context of implicit computational complexity theory.
Both techniques have been extensively employed, however separately. 
Essentially relying on sharing, the invariance of the unitary cost model in various 
rewriting based models of computation, 
e.g.\ the $\lambda$-calculus~\cite{ADL:RTA:12,LM:LMCS:12,AD:CSL:14} and 
term rewrite systems~\cite{LM:FOPARA:09,AM:RTA:10} could be proven. 
Several works (e.g.~\cite{Marion:IC:03,BMM:TCS:11,BDM:MSCS:12}) rely on memoization, 
employing a measure close to our notion of memoized runtime complexity. 
None of these works integrate sharing, instead, inputs are either restricted to strings 
or dedicated bounds on the size of intermediate values have to be imposed.  
We are confident that our second result is readily applicable to resolve such restrictions.


\bibliographystyle{plainnat}

\end{document}